\documentclass{article}%
\usepackage{amsfonts}
\usepackage{amssymb}
\usepackage{amsmath}
\usepackage{graphicx}%
\providecommand{\U}[1]{\protect\rule{.1in}{.1in}}

\newtheorem{theorem}{Theorem}

\newtheorem{example}[theorem]{Example}

\newenvironment{proof}[1][Proof]{\noindent\textbf{#1.} }{\ \rule{0.5em}{0.5em}}
\begin{document}

\title{New Logical Foundations for Quantum Information Theory: Introduction to
Quantum Logical Information Theory}
\author{David Ellerman\\University of California at Riverside}
\maketitle

\begin{abstract}
\noindent Logical information theory is the quantitative version of the logic
of partitions just as logical probability theory is the quantitative version
of the dual Boolean logic of subsets. The resulting notion of information is
about distinctions, differences, and distinguishability, and is formalized
using the distinctions (`dits') of a partition (a pair of points distinguished
by the partition). All the definitions of simple, joint, conditional, and
mutual entropy of Shannon information theory are derived by a uniform
transformation from the corresponding definitions at the logical level.

The purpose of this paper is to give the direct generalization to quantum
logical information theory that similarly focuses on the pairs of eigenstates
distinguished by an observable, i.e., \textit{qudits of an observable}. The
fundamental theorem for quantum logical entropy and measurement establishes a
direct quantitative connection between the increase in quantum logical entropy
due to a projective measurement and the eigenstates (cohered together in the
pure superposition state being measured) that are distinguished by the
measurement (decohered in the post-measurement mixed state). Both the
classical and quantum versions of logical entropy have simple interpretations
as \textquotedblleft two-draw\textquotedblright\ probabilities for
distinctions. The conclusion is that quantum logical entropy is the simple and
natural notion of information for quantum information theory focusing on the
distinguishing of quantum states.

\end{abstract}
\tableofcontents

\section{Duality of Subsets and Partitions}

The foundations for classical and quantum logical information theory are built
on the logic of partitions--which is dual (in the category-theoretic sense) to
the usual Boolean logic of subsets. F. William Lawvere called a subset or, in
general, a subobject a \textquotedblleft part\textquotedblright\ and then
noted: \textquotedblleft The dual notion (obtained by reversing the arrows) of
`part' is the notion of \textit{partition}.\textquotedblright\ \cite[p.
85]{law:sfm} That suggests that the Boolean logic of subsets has a dual logic
of partitions (\cite{ell:partitions}, \cite{ell:intropartlogic}).

This duality can be most simply illustrated using a set function
$f:X\rightarrow Y$. The image $f\left(  X\right)  $ is a \textit{subset} of
the codomain $Y$ and the inverse-image or coimage\footnote{In category theory,
the duality between subobject-type constructions (e.g., limits) and
quotient-object-type constructions (e.g., colimits) is often indicated by
adding the prefix \textquotedblleft co-\textquotedblright\ to the latter.
Hence the usual Boolean logic of \textquotedblleft images\textquotedblright%
\ has the dual logic of \textquotedblleft coimages.\textquotedblright}
$f^{-1}\left(  Y\right)  $ is a \textit{partition} on the domain $X$--where a
\textit{partition} $\pi=\left\{  B_{1},...,B_{I}\right\}  $ on a set $U$ is a
set of subsets or blocks\ $B_{i}$ that are mutually disjoint and jointly
exhaustive ($\cup_{i}B_{i}=U$). But the duality runs deeper than between
subsets and partitions. The dual to the notion of an \textquotedblleft
element\textquotedblright\ (an `it') of a subset is the notion of a
\textquotedblleft distinction\textquotedblright\ (a `dit') of a partition,
where $\left(  u,u^{\prime}\right)  \in U\times U$ is a \textit{distinction}
or \textit{dit} of $\pi$ if the two elements are in different blocks. Let
$\operatorname*{dit}\left(  \pi\right)  \subseteq U\times U$ be the set of
distinctions or \textit{ditset} of $\pi$. Similarly an \textit{indistinction}
or \textit{indit} of $\pi$ is a pair $\left(  u,u^{\prime}\right)  \in U\times
U$ in the same block of $\pi$. Let $\operatorname*{indit}\left(  \pi\right)
\subseteq U\times U$ be the set of indistinctions or \textit{inditset }of
$\pi$. Then $\operatorname*{indit}\left(  \pi\right)  $ is the equivalence
relation associated with $\pi$ and $\operatorname*{dit}\left(  \pi\right)
=U\times U-\operatorname*{indit}\left(  \pi\right)  $ is the complementary
binary relation that might be called a\textit{ partition relation} or an
\textit{apartness relation}. The notions of a distinction and indistinction of
a partition are illustrated in Figure 1.%

\begin{center}
\includegraphics[
height=1.4209in,
width=2.4794in
]%
{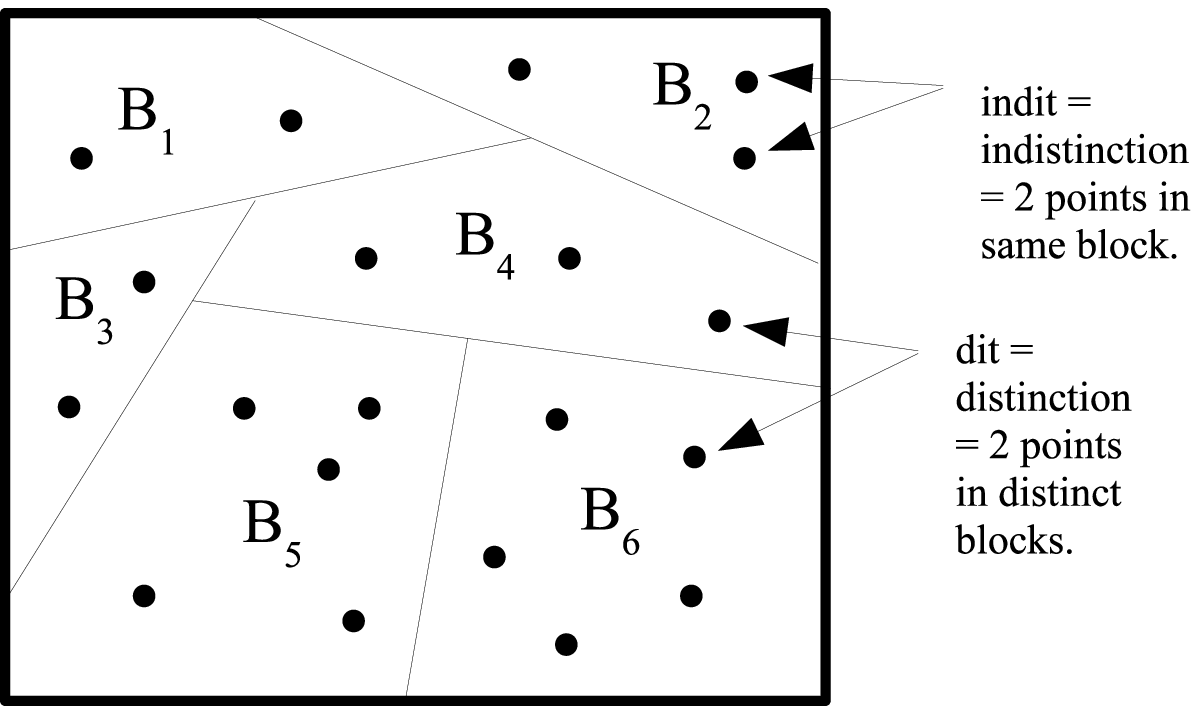}%
\end{center}

\begin{center}
Figure 1: Distinctions and indistinctions of a partition.
\end{center}

The Boolean logic of subsets is usually treated in modern texts solely in
terms of the special case of \textquotedblleft propositional
logic.\textquotedblright\ Given a formula $\Phi\left(  \pi,\sigma,...\right)
$ composed with Boolean operations (e.g., $\vee$, $\wedge$, $\Longrightarrow$,
$\emptyset$) on the atomic variables $\pi$, $\sigma$, ..., then a
\textit{Boolean tautology} should be defined as a formula such that no matter
what subsets of the nonempty universe $U$ are substituted for the atomic
variables, then the whole formula evaluates (using the corresponding set
operations) to the universe $U$, the top of the power-set Boolean algebra
$\wp\left(  U\right)  $. It is then a theorem (not a definition) that the same
set of valid formulas is obtained if one only considers the one-element
universe $U=1$, in which case it is convenient to interpret the variables and
formulas as being propositions. Most modern texts just start with this
propositional special case and \textit{define} a valid formula as a truth
table tautology, i.e., as a formula such that no matter what subsets
$\emptyset,1$ of the universe $1$ are substituted for the atomic variables,
the whole formula will evaluate to the universe $1$. This neglect of the
general Boolean logic of subsets in favor of the propositional special case is
one of the reasons for the long delay in developing the dual logic of
partitions--since propositions, unlike subsets, don't have a
category-theoretic dual \cite{ell:partitions}.

The algebra associated with the subsets $S\subseteq U$ is the power-set
Boolean algebra $\wp\left(  U\right)  $ of subsets of $U$ with the partial
order as the inclusion of elements. The corresponding algebra of partitions
$\pi$ on $U$ is the \textit{partition algebra} $\prod\left(  U\right)  $
defined as follows:

\begin{itemize}
\item the lattice \textit{partial order} $\sigma\preceq\pi$ of partitions
$\sigma=\left\{  C_{1},...,C_{J}\right\}  $ and $\pi=\left\{  B_{1}%
,...,B_{I}\right\}  $ holds when $\pi$ \textit{refines} $\sigma$ in the sense
that for every block $B_{i}\in\pi$ there is a block $C_{j}\in\sigma$ such that
$B_{i}\subseteq C_{j}$, or, equivalently, using the element-distinction (`its'
\& `dits') pairing, the partial order is the inclusion of distinctions:
$\sigma\preceq\pi$ if and only if (iff) $\operatorname*{dit}\left(
\sigma\right)  \subseteq\operatorname*{dit}\left(  \pi\right)  $;

\item the minimum or bottom partition is the \textit{indiscrete partition} (or
blob) $\mathbf{0}=\left\{  U\right\}  $ with one block consisting of all of
$U$;

\item the maximum or top partition is the \textit{discrete partition}
$\mathbf{1}=\left\{  \left\{  u\right\}  \right\}  _{u\in U}$ consisting of
singleton blocks;

\item the \textit{join} $\pi\vee\sigma$ is the partition whose blocks are the
non-empty intersections $B_{i}\cap C_{j}$ of blocks of $\pi$ and blocks of
$\sigma$, or, equivalently, using the element-distinction pairing,
$\operatorname*{dit}\left(  \pi\vee\sigma\right)  =\operatorname*{dit}\left(
\pi\right)  \cup\operatorname*{dit}\left(  \sigma\right)  $;

\item the \textit{meet} $\pi\wedge\sigma$ is the partition whose blocks are
the equivalence classes for the equivalence relation generated by: $u\sim
u^{\prime}$ if $u\in B_{i}\in\pi$, $u^{\prime}\in C_{j}\in\sigma$, and
$B_{i}\cap C_{j}\neq\emptyset$; and

\item $\sigma\Rightarrow\pi$ is the \textit{implication partition} whose
blocks are: (1) the singletons $\left\{  u\right\}  $ for $u\in B_{i}\in\pi$
if there is a $C_{j}\in\sigma$ such that $B_{i}\subseteq C_{j}$, or (2) just
$B_{i}\in\pi$ if there is no $C_{j}\in\sigma$ with $B_{i}\subseteq C_{j}$, so
that trivially: $\sigma\Rightarrow\pi=\mathbf{1}$ iff $\sigma\preceq\pi$.
\end{itemize}

The same formulas $\Phi\left(  \pi,\sigma,...\right)  $ can be interpreted as
subset formulas or partition formulas. A \textit{partition tautology} is
analogously defined as a formula such that no matter what partitions on any
$U$ ($\left\vert U\right\vert \geq2$) are substituted for the variables, the
whole formula will evaluate by the partition operations to the discrete
partition $\mathbf{1}$, the top of the partition algebra $\prod\left(
U\right)  $. For instance, \textit{modus ponens}, $\left(  \sigma\wedge\left(
\sigma\Rightarrow\pi\right)  \right)  \Rightarrow\pi$, is a partition
tautology. There is no $U$, analogous to $U=1$, such that a formula is a
partition tautology if and only if it always evaluates to $\mathbf{1}$ for
partitions on $U$ \cite[Proposition 1.18]{ell:partitions}.

There is a better way to connect subsets and partitions to propositions by
considering a generic element $u$ and a generic distinction $\left(
u,u^{\prime}\right)  $ (with $u\neq u^{\prime}$ understood). If a formula
$\Phi\left(  \pi,\sigma,...\right)  $ is construed as a subset formula, then
\textquotedblleft$u$ is an element of $\Phi\left(  \pi,\sigma,...\right)
$\textquotedblright\ [i.e., $u\in\Phi\left(  \pi,\sigma,...\right)  $] is the
corresponding proposition that is always true when $\Phi\left(  \pi
,\sigma,...\right)  $ is a Boolean tautology. If the formula $\Phi\left(
\pi,\sigma,...\right)  $ is construed as a partition formula, then the
corresponding proposition \textquotedblleft$\left(  u,u^{\prime}\right)  $ is
a distinction of $\Phi\left(  \pi,\sigma,...\right)  $\textquotedblright\ is
always true if $\Phi\left(  \pi,\sigma,...\right)  $ is a partition
tautology.\footnote{See \cite{ell:partitions} for how to use this
propositional connection to develop a consistent and complete system of
semantic tableaus for partition tautologies. For a given universe set $U$,
there is one logic of partitions on $U$ just as there is one Boolean logic of
subsets of $U$, and it does not seem to be directly related to Karl Svozil's
\textquotedblleft partition logics\textquotedblright\ \cite{svozil:qlogic}
which are certain sets of partitions that arise in the study of
quasi-classical models, such as Ron Wright's \textquotedblleft generalized urn
models\textquotedblright\ \cite{wright:genurnmodels}, for the quantum logic of
subspaces.} These results are summarized in Table 1 which illustrates the dual
relationship between the elements (`its') of a subset and the distinctions
(`dits') of a partition.

\begin{center}%
\begin{center}
\includegraphics[
height=2.6751in,
width=4.8158in
]%
{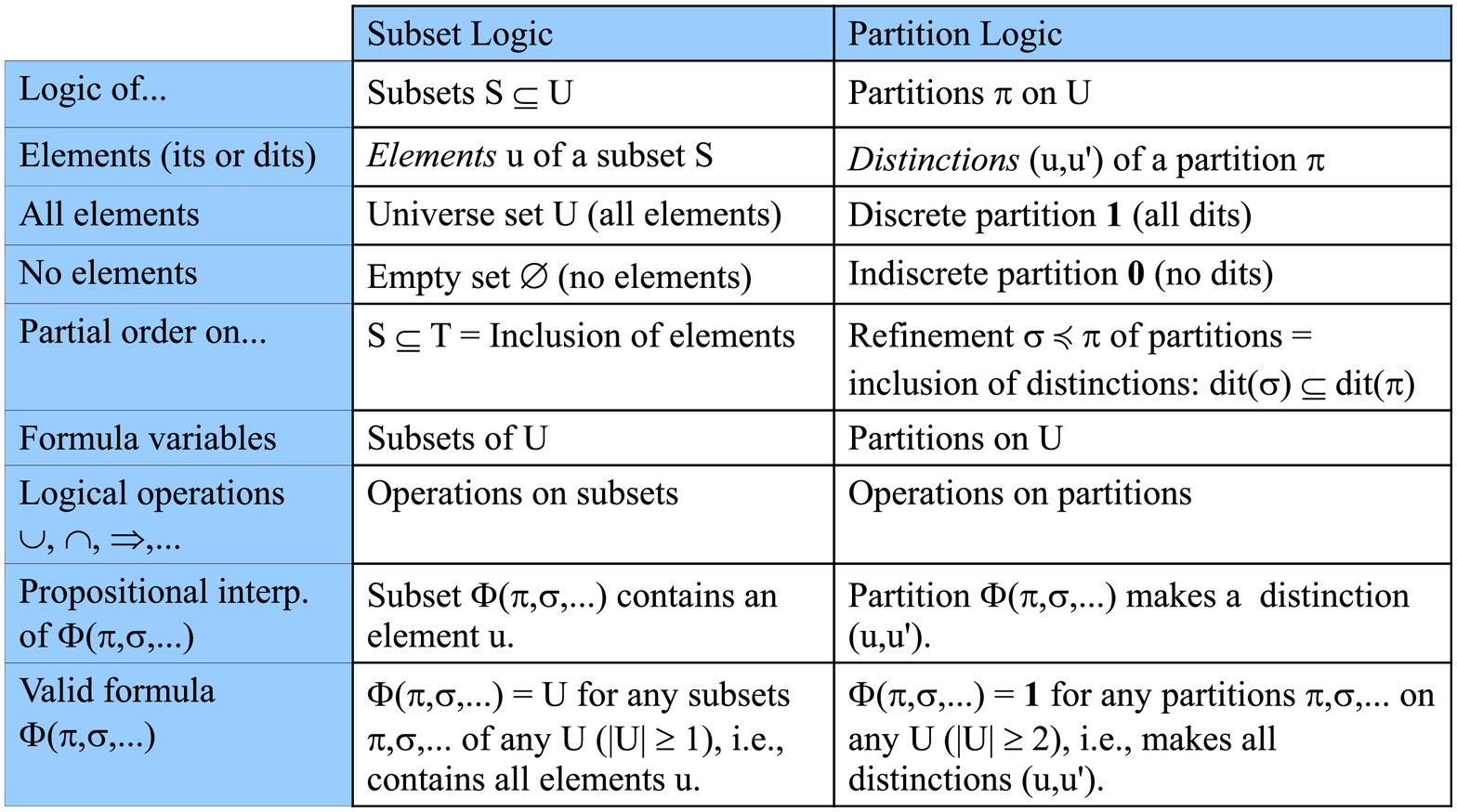}%
\end{center}

Table 1: Dual Logics: Boolean subset logic of subsets and partition logic.
\end{center}

\section{From the logic of partitions to logical information theory}

George Boole \cite{boole:lot} developed the quantitative version of his logic
of subsets by starting with the size or number of elements $\left\vert
S\right\vert $ in a subset $S\subseteq U$, which could then be normalized to
$\frac{\left\vert S\right\vert }{\left\vert U\right\vert }$ and given the
probabilistic interpretation as the probability that a randomly drawn element
from $U$ would be an element of $S$. The algebra of partitions $\pi$ on $U$ is
isomorphically represented by the algebra of ditsets $\operatorname*{dit}%
\left(  \pi\right)  \subseteq U\times U$, so the parallel quantitative
development of the logic of partitions would start with the size or number of
distinctions $\left\vert \operatorname*{dit}\left(  \pi\right)  \right\vert $
in a partition $\pi$ on $U$, which could then be normalized to $\frac
{\left\vert \operatorname*{dit}\left(  \pi\right)  \right\vert }{\left\vert
U\times U\right\vert }$. It has the probabilistic interpretation as the
probability that two randomly drawn elements from $U$ (with replacement) would
be a distinction of $\pi$.

In Gian-Carlo Rota's Fubini Lectures \cite{rota:fubini} (and in his lectures
at MIT), he remarked in view of duality between partitions and subsets that,
quantitatively, the \textquotedblleft lattice of partitions plays for
information the role that the Boolean algebra of subsets plays for size or
probability\textquotedblright\ \cite[p. 30]{kung-rota-yan:comb} or symbolically:

\begin{center}
$\frac{\text{probability}}{\text{subsets}}=\frac{\text{information}%
}{\text{partitions.}}$
\end{center}

\noindent Since \textquotedblleft Probability is a measure on the Boolean
algebra of events\textquotedblright\ that gives quantitatively the
\textquotedblleft intuitive idea of the size of a set\textquotedblright, we
may ask by \textquotedblleft analogy\textquotedblright\ for some measure to
capture a property for a partition like \textquotedblleft what size is to a
set.\textquotedblright\ Rota goes on to ask:

\begin{quotation}
\noindent How shall we be led to such a property? We have already an inkling
of what it should be: it should be a measure of information provided by a
random variable. Is there a candidate for the measure of the amount of
information? \cite[p. 67]{rota:fubini}
\end{quotation}

\noindent We have just seen that the parallel development suggests the
normalized number of distinctions of a partition as \textquotedblleft the
measure of the amount of information\textquotedblright.

\section{The logical theory of information}

Andrei Kolmogorov has suggested that information theory should start with
sets, not probabilities.

\begin{quotation}
\noindent Information theory must precede probability theory, and not be based
on it. By the very essence of this discipline, the foundations of information
theory have a finite combinatorial character. \cite[p. 39]{kolmogor:combfound}
\end{quotation}

The notion of information-as-distinctions does start with the \textit{set of
distinctions}, the \textit{information set}, of a partition $\pi=\left\{
B_{1},...,B_{I}\right\}  $ on a finite set $U$ where that set of distinctions
(dits) is:

\begin{center}
$\operatorname*{dit}\left(  \pi\right)  =\left\{  \left(  u,u^{\prime}\right)
:\exists B_{i},B_{i^{\prime}}\in\pi,B_{i}\neq B_{i^{\prime}},u\in
B_{i},u^{\prime}\in B_{i^{\prime}}\right\}  $.
\end{center}

The normalized size of a subset is the logical probability of the event, and
the normalized size of the ditset of a partition is, in the sense of measure
theory, \textquotedblleft the measure of the amount of
information\textquotedblright\ in a partition. Thus we define the
\textit{logical entropy} of a partition $\pi=\left\{  B_{1,}...,B_{I}\right\}
$, denoted $h\left(  \pi\right)  $, as the size of the ditset
$\operatorname*{dit}\left(  \pi\right)  \subseteq U\times U$ normalized by the
size of $U\times U$:

\begin{center}
$h\left(  \pi\right)  =\frac{\left\vert \operatorname*{dit}\left(  \pi\right)
\right\vert }{\left\vert U\times U\right\vert }=\sum_{\left(  u_{j}%
,u_{k}\right)  \in\operatorname*{dit}\left(  \pi\right)  }\frac{1}{\left\vert
U\right\vert }\frac{1}{\left\vert U\right\vert }$

Logical entropy of $\pi$ (equiprobable case).
\end{center}

\noindent This is just the product probability measure on $U\times U$ of the
equiprobable or uniform probability distribution on $U$ applied to the
information set or ditset $\operatorname*{dit}\left(  \pi\right)  $. The
inditset of $\pi$ is $\operatorname*{indit}\left(  \pi\right)  =\cup_{i=1}%
^{I}\left(  B_{i}\times B_{i}\right)  $ so where $\Pr\left(  B_{i}\right)
=\frac{|B_{i}|}{\left\vert U\right\vert }$ in the equiprobable case, we have:

\begin{center}
$h\left(  \pi\right)  =\frac{\left\vert \operatorname*{dit}\left(  \pi\right)
\right\vert }{\left\vert U\times U\right\vert }=\frac{\left\vert U\times
U\right\vert -\sum_{i=1}^{I}\left\vert B_{i}\times B_{i}\right\vert
}{\left\vert U\times U\right\vert }=1-\sum_{i=1}^{I}\left(  \frac{\left\vert
B_{i}\right\vert }{\left\vert U\right\vert }\right)  ^{2}=1-\sum_{i=1}^{I}%
\Pr\left(  B_{i}\right)  ^{2}$.
\end{center}

\noindent In two independent draws from $U$, the probability of getting a
distinction of $\pi$ is the probability of not getting an indistinction.

Given any probability measure $p:U\rightarrow\lbrack0,1]$ on $U=\left\{
u_{1},...,u_{n}\right\}  $ which defines $p_{i}=p\left(  u_{i}\right)  $ for
$i=1,...,n$, the\textit{\ product measure} $p\times p:U\times U\rightarrow
\left[  0,1\right]  $ has for any binary relation $R\subseteq U\times U$ the
value of:

\begin{center}
$p\times p\left(  R\right)  =\sum_{\left(  u_{i},u_{j}\right)  \in R}p\left(
u_{i}\right)  p\left(  u_{j}\right)  =\sum_{\left(  u_{i},u_{j}\right)  \in
R}p_{i}p_{j}$.
\end{center}

The \textit{logical entropy} of $\pi$ in general is the product-probability
measure of its ditset $\operatorname*{dit}\left(  \pi\right)  \subseteq
U\times U$, where $\Pr\left(  B\right)  =\sum_{u\in B}p\left(  u\right)  $:

\begin{center}
$h\left(  \pi\right)  =p\times p\left(  \operatorname*{dit}\left(  \pi\right)
\right)  =\sum_{\left(  u_{i},u_{j}\right)  \in\operatorname*{dit}\left(
\pi\right)  }p_{i}p_{j}=1-\sum_{B\in\pi}\Pr\left(  B\right)  ^{2}$.
\end{center}

There are two stages in the development of logical information. Before the
introduction of any probabilities, the information set of a partition $\pi$ on
$U$ is its ditset $\operatorname*{dit}\left(  \pi\right)  $. Then given a
probability measure $p:U\rightarrow\left[  0,1\right]  $ on $U$, the logical
entropy of the partition is just the product measure on the ditset, i.e.,
$h\left(  \pi\right)  =p\times p\left(  \operatorname*{dit}\left(  \pi\right)
\right)  $. The standard interpretation of $h\left(  \pi\right)  $ is the
two-draw probability of getting a distinction of the partition $\pi$--just as
$\Pr\left(  S\right)  $ is the one-draw probability of getting an element of
the subset-event $S$.

\section{Compound logical entropies}

The compound notions of logical entropy are also developed in two stages,
first as sets and then, given a probability distribution, as two-draw
probabilities. After observing the similarity between the formulas holding for
the compound Shannon entropies and the Venn diagram formulas that hold for any
measure (in the sense of measure theory), the information theorist, Lorne L.
Campbell, remarked in 1965 that the similarity:

\begin{quotation}
\noindent suggests the possibility that $H\left(  \alpha\right)  $ and
$H\left(  \beta\right)  $ are measures of sets, that $H\left(  \alpha
,\beta\right)  $ is the measure of their union, that $I\left(  \alpha
,\beta\right)  $ is the measure of their intersection, and that $H\left(
\alpha|\beta\right)  $ is the measure of their difference. The possibility
that $I\left(  \alpha,\beta\right)  $ is the entropy of the \textquotedblleft
intersection\textquotedblright\ of two partitions is particularly interesting.
This \textquotedblleft intersection,\textquotedblright\ if it existed, would
presumably contain the information common to the partitions $\alpha$ and
$\beta$. \cite[p. 113]{camp:meas}
\end{quotation}

\noindent Yet, there is no such interpretation of the Shannon entropies as
measures of sets, but the logical entropies precisely fulfill Campbell's
suggestion (with the \textquotedblleft intersection\textquotedblright\ of two
partitions being the intersection of their ditsets). Moreover, there is a
uniform requantifying transformation (see next section) that obtains all the
Shannon definitions from the logical definitions and explains how the Shannon
entropies can satisfy the Venn diagram formulas (e.g., as a mnemonic) while
not being defined by a measure on sets.

Given partitions $\pi=\left\{  B_{1},...,B_{I}\right\}  $ and $\sigma=\left\{
C_{1},...,C_{J}\right\}  $ on $U$, the \textit{joint information set} is the
union of the ditsets which is also the \textit{ditset for their join} is:
$\operatorname*{dit}\left(  \pi\right)  \cup\operatorname*{dit}\left(
\sigma\right)  =\operatorname*{dit}\left(  \pi\vee\sigma\right)  \subseteq
U\times U$. Given probabilities $p=\left\{  p_{1},...,p_{n}\right\}  $ on $U$,
the \textit{joint logical entropy} is the product probability measure on the
union of ditsets:

\begin{center}
$h\left(  \pi,\sigma\right)  =h\left(  \pi\vee\sigma\right)  =p\times p\left(
\operatorname*{dit}\left(  \pi\right)  \cup\operatorname*{dit}\left(
\sigma\right)  \right)  =1-\sum_{i,j}\Pr\left(  B_{i}\cap C_{j}\right)  ^{2}$.
\end{center}

\noindent The information set for the \textit{conditional logical entropy}
$h\left(  \pi|\sigma\right)  $ is the difference of ditsets, and thus that
logical entropy is:

\begin{center}
$h\left(  \pi|\sigma\right)  =p\times p\left(  \operatorname*{dit}\left(
\pi\right)  -\operatorname*{dit}\left(  \sigma\right)  \right)  =h\left(
\pi,\sigma\right)  -h\left(  \sigma\right)  $.
\end{center}

\noindent The information set for the \textit{logical mutual information}
$m\left(  \pi,\sigma\right)  \ $is the intersection of ditsets, so that
logical entropy is:

\begin{center}
$m\left(  \pi,\sigma\right)  =p\times p\left(  \operatorname*{dit}\left(
\pi\right)  \cap\operatorname*{dit}\left(  \sigma\right)  \right)  =h\left(
\pi,\sigma\right)  -h\left(  \pi|\sigma\right)  -h\left(  \sigma|\pi\right)
=h\left(  \pi\right)  +h\left(  \sigma\right)  -h\left(  \pi,\sigma\right)  $.
\end{center}

Since all the logical entropies are the values of a measure $p\times p:U\times
U\rightarrow\left[  0,1\right]  $ on subsets of $U\times U$, they
automatically satisfy the usual Venn diagram relationships.

\begin{center}%
\begin{center}
\includegraphics[
height=1.5731in,
width=1.7798in
]%
{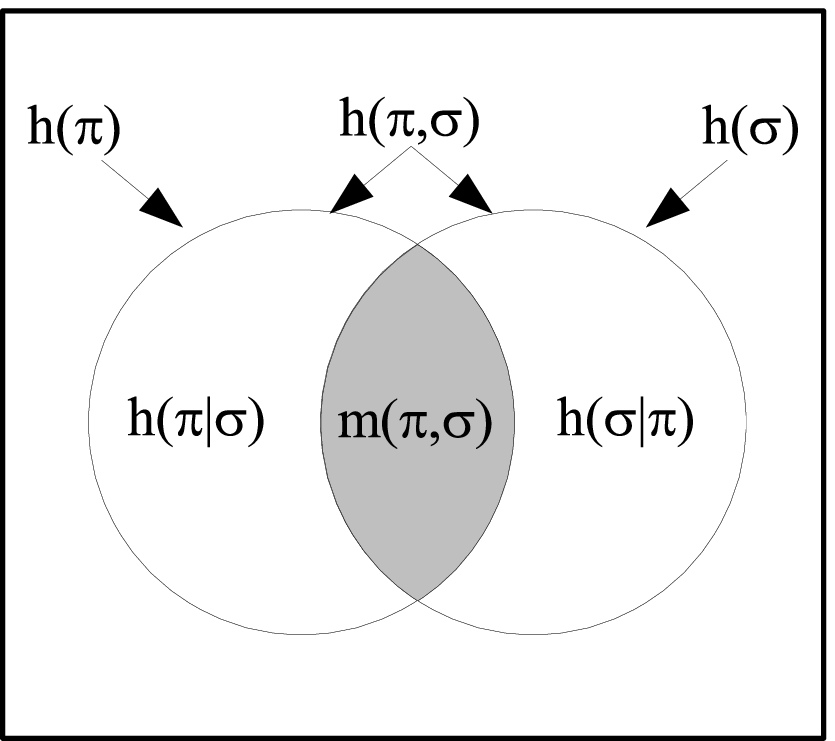}%
\end{center}

Figure 2: Venn diagram for logical entropies

as values of a probability measure $p\times p$ on $U\times U$
\end{center}

At the level of information sets (w/o probabilities), we have the
\textit{information algebra} $I\left(  \pi,\sigma\right)  $ which is the
Boolean subalgebra of $\wp\left(  U\times U\right)  $ generated by ditsets and
their complements.

\section{Deriving the Shannon entropies from the logical entropies}

Instead of being defined as the values of a measure, the usual notions of
simple and compound entropy `burst forth fully formed from the forehead' of
Claude Shannon \cite{shannon:comm} already satisfying the standard Venn
diagram relationships.\footnote{One author surmised that \textquotedblleft
Shannon carefully contrived for this `accident' to occur\textquotedblright%
\ \cite[p. 153]{rozeboom:partials}.} Since the Shannon entropies are not the
values of a measure, many authors have pointed out that these Venn diagram
relations for the Shannon entropies can only be taken as \textquotedblleft
analogies\textquotedblright\ or \textquotedblleft mnemonics\textquotedblright%
\ (\cite{camp:meas}; \cite{abramson:it}). Logical information theory explains
this situation since all the Shannon definitions of simple, joint,
conditional, and mutual information can be obtained by a uniform requantifying
transformation from the corresponding logical definitions, and the
transformation preserves the Venn diagram relationships.

This transformation is possible since the logical and Shannon notions of
entropy can be seen as two different ways to quantify distinctions--and thus
\textit{both} theories are based on the foundational idea of
\textit{information-as-distinctions}.

Consider the canonical case of $n$ equiprobable elements, $p_{i}=\frac{1}{n}$.
The logical entropy of $\mathbf{1}=\left\{  B_{1},...,B_{n}\right\}  $ where
$B_{i}=\left\{  u_{i}\right\}  $ with $p=\left\{  \frac{1}{n},...,\frac{1}%
{n}\right\}  $ is:

\begin{center}
$\frac{\left\vert U\times U-\Delta\right\vert }{\left\vert U\times
U\right\vert }=\frac{n^{2}-n}{n^{2}}=1-\frac{1}{n}=1-\Pr\left(  B_{i}\right)
$.
\end{center}

\noindent The normalized number of distinctions or `dit-count' of the discrete
partition $\mathbf{1}$ is $1-\frac{1}{n}=1-\Pr\left(  B_{i}\right)  $. The
general case of logical entropy for any $\pi=\left\{  B_{1},...,B_{I}\right\}
$ is the average of the dit-counts $1-\Pr\left(  B_{i}\right)  $ for the
canonical cases:

\begin{center}
$h\left(  \pi\right)  =\sum_{i}\Pr\left(  B_{i}\right)  \left(  1-\Pr\left(
B_{i}\right)  \right)  $.
\end{center}

In the canonical case of $2^{n}$ equiprobable elements, the minimum number of
binary partitions (\textquotedblleft yes-or-no questions\textquotedblright\ or
\textquotedblleft bits\textquotedblright)\ whose join is the discrete
partition $\mathbf{1}=\left\{  B_{1},...,B_{2^{n}}\right\}  $ with $\Pr\left(
B_{i}\right)  =\frac{1}{2^{n}}$, i.e., that it takes to uniquely
\textit{encode} each distinct element, is $n$, so the Shannon-Hartley entropy
\cite{hart:ti} is the canonical bit-count:

\begin{center}
$n=\log_{2}\left(  2^{n}\right)  =\log_{2}\left(  \frac{1}{1/2^{n}}\right)
=\log_{2}\left(  \frac{1}{\Pr\left(  B_{i}\right)  }\right)  $.
\end{center}

\noindent The general case Shannon entropy is the average of these canonical
bit-counts $\log_{2}\left(  \frac{1}{\Pr\left(  B_{i}\right)  }\right)  $:

\begin{center}
$H\left(  \pi\right)  =\sum_{i}\Pr\left(  B_{i}\right)  \log_{2}\left(
\frac{1}{\Pr\left(  B_{i}\right)  }\right)  $.
\end{center}

The \textit{Dit-Bit Transform} essentially replaces the canonical dit-counts
by the canonical bit-counts. First express any logical entropy concept
(simple, joint, conditional, or mutual) as an average of canonical dit-counts
$1-\Pr\left(  B_{i}\right)  $, and then substitute the canonical bit-count
$\log\left(  \frac{1}{\Pr\left(  B_{i}\right)  }\right)  $ to obtain the
corresponding formula as defined by Shannon. Table 2 gives examples of the
dit-bit transform.

\begin{center}%
\begin{center}
\includegraphics[
height=1.8911in,
width=4.2071in
]%
{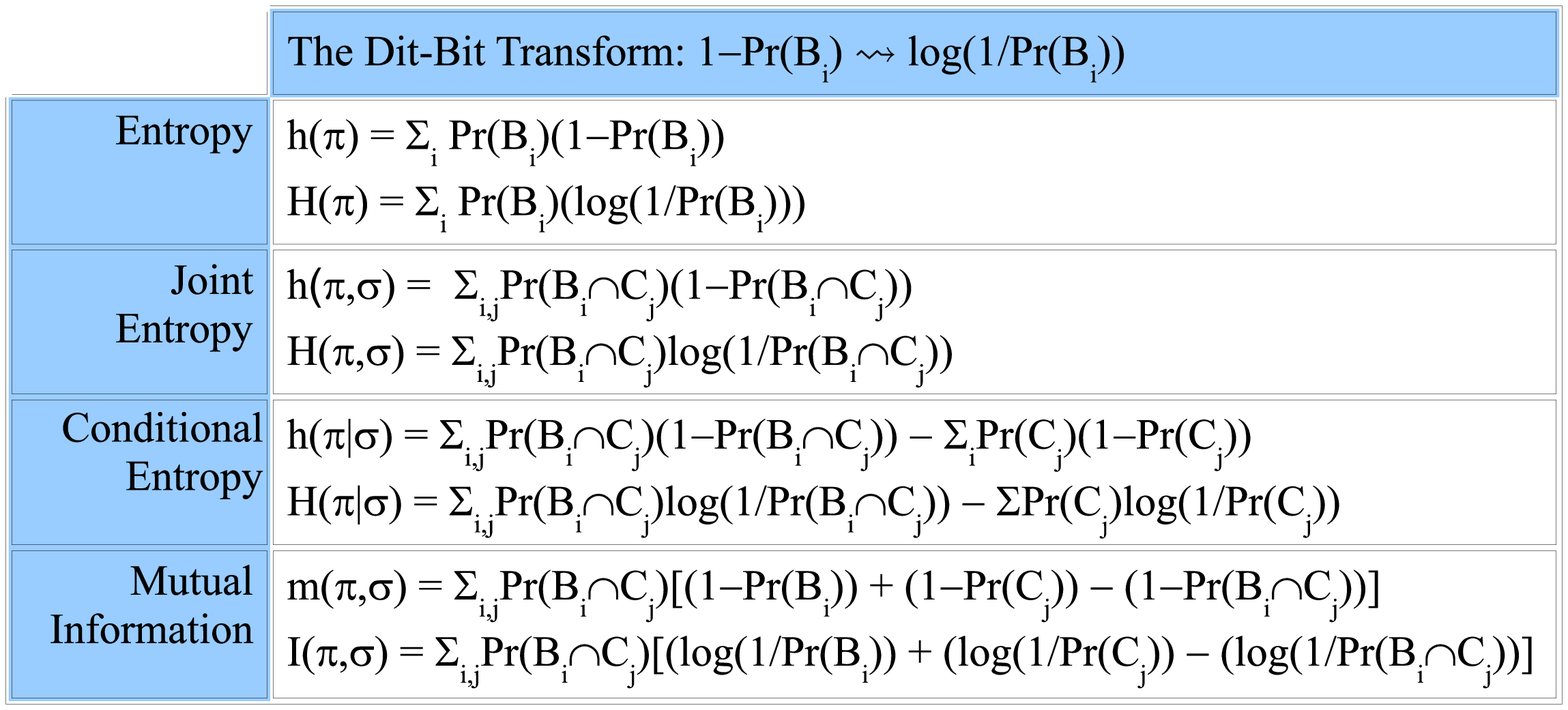}%
\end{center}

Table 2: Summary of the dit-bit transform
\end{center}

For instance,

\begin{center}
$h\left(  \pi|\sigma\right)  =h\left(  \pi,\sigma\right)  -h\left(
\sigma\right)  =\sum_{i,j}\Pr\left(  B_{i}\cap C_{j}\right)  \left[
1-\Pr\left(  B_{i}\cap C_{j}\right)  \right]  -\sum_{j}\Pr\left(
C_{j}\right)  \left[  1-\Pr\left(  C_{j}\right)  \right]  $
\end{center}

\noindent is the expression for $h\left(  \pi|\sigma\right)  $ as an average
over $1-\Pr\left(  B_{i}\cap C_{j}\right)  $ and $1-\Pr\left(  C_{j}\right)
$, so applying the dit-bit transform gives:

\begin{center}
$\sum_{i,j}\Pr\left(  B_{i}\cap C_{j}\right)  \log\left(  1/\Pr\left(
B_{i}\cap C_{j}\right)  \right)  -\sum_{j}\Pr\left(  C_{j}\right)  \log\left(
1/\Pr\left(  C_{j}\right)  \right)  =H\left(  \pi,\sigma\right)  -H\left(
\sigma\right)  =H\left(  \pi|\sigma\right)  $.
\end{center}

The dit-bit transform is linear in the sense of preserving plus and minus, so
the Venn diagram formulas, e.g., $h\left(  \pi,\sigma\right)  =h\left(
\sigma\right)  +h\left(  \pi|\sigma\right)  $, automatically satisfied by
logical entropy since it is a measure, carry over to Shannon entropy, e.g.,
$H\left(  \pi,\sigma\right)  =H\left(  \sigma\right)  +H\left(  \pi
|\sigma\right)  $, in spite of it not being a measure (in the sense of measure theory):

\begin{center}
$%
{\includegraphics[
height=1.5394in,
width=1.7659in
]%
{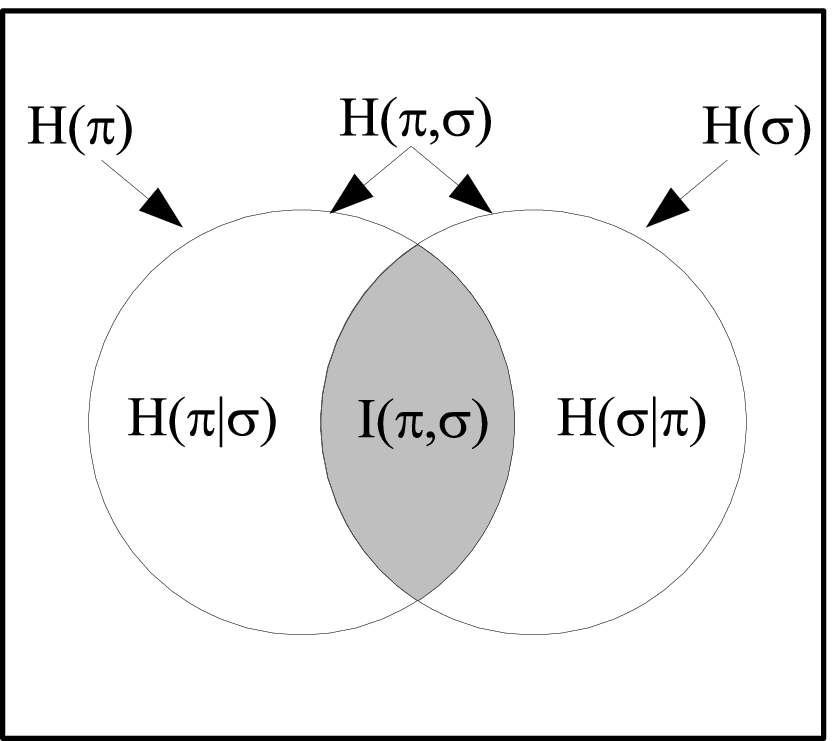}%
}%
$

Figure 3:\ Venn diagram mnemonic for Shannon entropies
\end{center}

\section{Logical entropy via density matrices}

The transition to quantum logical entropy is facilitated by reformulating the
logical theory in terms of density matrices. Let $U=\left\{  u_{1}%
,...,u_{n}\right\}  $ be the sample space with the point probabilities
$p=\left(  p_{1},...,p_{n}\right)  $. An event $S\subseteq U$ has the
probability $\Pr\left(  S\right)  =\sum_{u_{j}\in S}p_{j}$.

For any event $S$ with $\Pr\left(  S\right)  >0$, let

\begin{center}
$\left\vert S\right\rangle =\frac{1}{\sqrt{\Pr\left(  S\right)  }}(\chi
_{S}\left(  u_{1}\right)  \sqrt{p_{1}},...,\chi_{S}\left(  u_{n}\right)
\sqrt{p_{n}})^{t}$
\end{center}

\noindent(the superscript $t$ indicates transpose) which is a normalized
column vector in $%
\mathbb{R}
^{n}$ where $\chi_{S}:U\rightarrow\left\{  0,1\right\}  $ is the
characteristic function for $S$, and let $\left\langle S\right\vert $ be the
corresponding row vector. Since $\left\vert S\right\rangle $ is normalized,
$\left\langle S|S\right\rangle =1$. Then the \textit{density matrix}
representing the event $S$ is the $n\times n$ symmetric real matrix:

\begin{center}
$\rho\left(  S\right)  =\left\vert S\right\rangle \left\langle S\right\vert
=\left\{
\begin{array}
[c]{c}%
\frac{1}{\Pr\left(  S\right)  }\sqrt{p_{j}p_{k}}\text{ for }u_{j},u_{k}\in S\\
0\text{ otherwise}%
\end{array}
\right.  $.
\end{center}

\noindent Then $\rho\left(  S\right)  ^{2}=\left\vert S\right\rangle
\left\langle S|S\right\rangle \left\langle S\right\vert =\rho\left(  S\right)
$ so borrowing language from quantum mechanics, $\rho\left(  S\right)  $ is
said to be a \textit{pure state} density matrix.

Given any partition $\pi=\left\{  B_{1},...,B_{I}\right\}  $ on $U$, its
density matrix is the average of the block density matrices:

\begin{center}
$\rho\left(  \pi\right)  =\sum_{i}\Pr\left(  B_{i}\right)  \rho\left(
B_{i}\right)  $.
\end{center}

Then $\rho\left(  \pi\right)  $ represents the \textit{mixed state},
experiment, or lottery where the event $B_{i}$ occurs with probability
$\Pr\left(  B_{i}\right)  $. A little calculation connects the logical entropy
$h\left(  \pi\right)  $ of a partition with the density matrix treatment:

\begin{center}
$h\left(  \pi\right)  =1-\sum_{i=1}^{I}\Pr\left(  B_{i}\right)  ^{2}%
=1-\operatorname*{tr}\left[  \rho\left(  \pi\right)  ^{2}\right]  =h\left(
\rho\left(  \pi\right)  \right)  $
\end{center}

\noindent where $\rho\left(  \pi\right)  ^{2}$ is substituted for $\Pr\left(
B_{i}\right)  ^{2}$ and the trace is substituted for the summation.

\begin{example}
For the throw of a fair die, $U=\left\{  u_{1},u_{3},u_{5},u_{2},u_{4}%
,u_{6}\right\}  $ (note the odd faces ordered before the even faces in the
matrix rows and columns) where $u_{j}$ represents the number $j$ coming up,
the density matrix $\rho\left(  \mathbf{0}\right)  $ is the \textquotedblleft
pure state\textquotedblright\ $6\times6$ matrix with each entry being
$\frac{1}{6}$.
\end{example}

\begin{center}
$\rho\left(  \mathbf{0}\right)  =%
\begin{bmatrix}
1/6 & 1/6 & 1/6 & 1/6 & 1/6 & 1/6\\
1/6 & 1/6 & 1/6 & 1/6 & 1/6 & 1/6\\
1/6 & 1/6 & 1/6 & 1/6 & 1/6 & 1/6\\
1/6 & 1/6 & 1/6 & 1/6 & 1/6 & 1/6\\
1/6 & 1/6 & 1/6 & 1/6 & 1/6 & 1/6\\
1/6 & 1/6 & 1/6 & 1/6 & 1/6 & 1/6
\end{bmatrix}%
\begin{array}
[c]{c}%
u_{1}\\
u_{3}\\
u_{5}\\
u_{2}\\
u_{4}\\
u_{6}%
\end{array}
$.
\end{center}

The nonzero off-diagonal entries represent indistinctions or indits of the
partition $\mathbf{0}$, or in quantum terms, \textquotedblleft
coherences,\textquotedblright\ where all $6$ \textquotedblleft
eigenstates\textquotedblright\ cohere together in a pure \textquotedblleft
superposition\textquotedblright\ state. All pure states have logical entropy
of zero, i.e., $h\left(  \mathbf{0}\right)  =0$ (i.e., no dits) since
$\operatorname*{tr}\left[  \rho\right]  =1$ for any density matrix so if
$\rho\left(  \mathbf{0}\right)  ^{2}=\rho\left(  \mathbf{0}\right)  $, then
$\operatorname*{tr}\left[  \rho\left(  \mathbf{0}\right)  ^{2}\right]
=\operatorname*{tr}\left[  \rho\left(  \mathbf{0}\right)  \right]  =1$ and
$h\left(  0\right)  =1-\operatorname*{tr}\left[  \rho\left(  \mathbf{0}%
\right)  ^{2}\right]  =0$. The logical operation of classifying
undistinguished entities (like the six faces of the die before a throw to
determine a face up) by a numerical attribute makes distinctions between the
entities with different numerical values of the attribute. It is the classical
operation corresponding to the quantum operation of `measurement' of a
superposition state by an observable.

\begin{example}
[continued]Now classify or \textquotedblleft measure\textquotedblright\ the
elements by the parity-of-the-face-up (odd or even) partition (observable)
$\pi=\left\{  B_{odd},B_{even}\right\}  =\left\{  \left\{  u_{1},u_{3}%
,u_{5}\right\}  ,\left\{  u_{2},u_{4},u_{6}\right\}  \right\}  $.
Mathematically, this is done by the L\"{u}ders mixture operation \cite[p.
279]{auletta:qm}, i.e., pre- and post-multiplying the density matrix
$\rho\left(  \mathbf{0}\right)  $ by $P_{odd}$ and by $P_{even}$, the
projections to the odd or even components, and summing the results:
\end{example}

\begin{center}
$P_{odd}\rho\left(  \mathbf{0}\right)  P_{odd}+P_{even}\rho\left(
\mathbf{0}\right)  P_{even}$

$=%
\begin{bmatrix}
1/6 & 1/6 & 1/6 & 0 & 0 & 0\\
1/6 & 1/6 & 1/6 & 0 & 0 & 0\\
1/6 & 1/6 & 1/6 & 0 & 0 & 0\\
0 & 0 & 0 & 0 & 0 & 0\\
0 & 0 & 0 & 0 & 0 & 0\\
0 & 0 & 0 & 0 & 0 & 0
\end{bmatrix}
+%
\begin{bmatrix}
0 & 0 & 0 & 0 & 0 & 0\\
0 & 0 & 0 & 0 & 0 & 0\\
0 & 0 & 0 & 0 & 0 & 0\\
0 & 0 & 0 & 1/6 & 1/6 & 1/6\\
0 & 0 & 0 & 1/6 & 1/6 & 1/6\\
0 & 0 & 0 & 1/6 & 1/6 & 1/6
\end{bmatrix}
$

$=%
\begin{bmatrix}
1/6 & 1/6 & 1/6 & 0 & 0 & 0\\
1/6 & 1/6 & 1/6 & 0 & 0 & 0\\
1/6 & 1/6 & 1/6 & 0 & 0 & 0\\
0 & 0 & 0 & 1/6 & 1/6 & 1/6\\
0 & 0 & 0 & 1/6 & 1/6 & 1/6\\
0 & 0 & 0 & 1/6 & 1/6 & 1/6
\end{bmatrix}
$

$=\frac{1}{2}\rho\left(  B_{odd}\right)  +\frac{1}{2}\rho\left(
B_{even}\right)  =\rho\left(  \pi\right)  $.
\end{center}

\begin{theorem}
[Fundamental]The increase in logical entropy, $h\left(  \rho\left(
\pi\right)  \right)  -h\left(  \rho\left(  \mathbf{0}\right)  \right)  $, due
to a L\"{u}ders mixture operation is the sum of amplitudes squared of the
non-zero off-diagonal entries of the beginning density matrix that are zeroed
in the final density matrix.
\end{theorem}

\begin{proof}
Since for any density matrix $\rho$, $\operatorname*{tr}\left[  \rho
^{2}\right]  =\sum_{i,j}\left\vert \rho_{ij}\right\vert ^{2}$\cite[p.
77]{fano:density}, we have: $h\left(  \rho\left(  \pi\right)  \right)
-h\left(  \rho\left(  \mathbf{0}\right)  \right)  =\left(
1-\operatorname*{tr}\left[  \rho\left(  \pi\right)  ^{2}\right]  \right)
-\left(  1-\operatorname*{tr}\left[  \rho\left(  \mathbf{0}\right)
^{2}\right]  \right)  =\operatorname*{tr}\left[  \rho\left(  \mathbf{0}%
\right)  ^{2}\right]  -\operatorname*{tr}\left[  \rho\left(  \pi\right)
^{2}\right]  =\sum_{i,j}\left(  \left\vert \rho_{ij}\left(  \mathbf{0}\right)
\right\vert ^{2}-\left\vert \rho_{ij}\left(  \pi\right)  \right\vert
^{2}\right)  $. If $\left(  u_{i},u_{i^{\prime}}\right)  \in
\operatorname*{dit}\left(  \pi\right)  $, then and only then are the
off-diagonal terms corresponding to $u_{i}$ and $u_{i^{\prime}}$ zeroed by the
L\"{u}ders operation.
\end{proof}

The fundamental theorem connects the concept of information-as-distinctions to
the process of `measurement' or classification which uses some attribute (like
parity in the example) or `observable' to make distinctions.

\begin{example}
[continued]In comparison with the matrix $\rho\left(  \mathbf{0}\right)  $ of
all entries $\frac{1}{6}$, the entries that got zeroed in the L\"{u}ders
operation $\rho\left(  \mathbf{0}\right)  \leadsto\rho\left(  \pi\right)  $
correspond to the distinctions created in the transition $\mathbf{0}=\left\{
\left\{  u_{1},...,u_{6}\right\}  \right\}  \leadsto\pi=\left\{  \left\{
u_{1},u_{3},u_{5}\right\}  ,\left\{  u_{2},u_{4},u_{6}\right\}  \right\}  $,
i.e., the odd-numbered faces were distinguished from the even-numbered faces
by the parity attribute. The increase in logical entropy = sum of the squares
of the off-diagonal elements that were zeroed = $h\left(  \pi\right)
-h\left(  \mathbf{0}\right)  =2\times9\times\left(  \frac{1}{6}\right)
^{2}=\frac{18}{36}=\frac{1}{2}$. The usual calculations of the two logical
entropies are: $h\left(  \pi\right)  =1-2\times\left(  \frac{1}{2}\right)
^{2}=\frac{1}{2}$ and $h\left(  \mathbf{0}\right)  =1-1=0$.
\end{example}

Since, in quantum mechanics, a projective measurement's effect on a density
matrix \textit{is} the L\"{u}ders mixture operation, that means that the
effects of the measurement are the above-described \textquotedblleft making
distinctions\textquotedblright\ by decohering or zeroing certain coherence
terms in the density matrix, and the sum of the absolute squares of the
coherences that were decohered is the increase in the logical entropy.

\section{Quantum Logical Information Theory: Commuting Observables}

The idea of information-as-distinctions carries over to quantum mechanics.

\begin{quotation}
\noindent\lbrack Information] is the notion of distinguishability abstracted
away from what we are distinguishing, or from the carrier of information.
...And we ought to develop a theory of information which generalizes the
theory of distinguishability to include these quantum properties...
.\ \cite[p. 155]{bennett:qinfo}
\end{quotation}

Let $F:V\rightarrow V$ be a self-adjoint operator (observable) on a
$n$-dimensional Hilbert space $V$ with the real eigenvalues $\phi_{1}%
,...,\phi_{I}$ and let $U=\left\{  u_{1},...,u_{n}\right\}  $ be an
orthonormal (ON) basis of eigenvectors of $F$. The quantum version of a dit, a
\textit{qudit}, is a pair of states definitely distinguishable by
\textit{some} observable\footnote{Any nondegenerate self-adjoint operator such
as $\sum_{k=1}^{n}kP_{\left[  u_{k}\right]  }$, where $P_{\left[
u_{k}\right]  }$ is the projection to the one-dimensional subspace generated
by $u_{k}$, will distinguish all the vectors in the orthonormal basis $U$%
.}--which is analogous classically to a pair $\left(  u,u^{\prime}\right)  $
of distinct elements of $U$ that are distinguishable by some partition (i.e.,
$\mathbf{1}$). In general, a \textit{qudit }is \textit{relativized to an
observable}--just as classically a distinction is a distinction \textit{of a
partition}. Then there is a set partition $\pi=\left\{  B_{i}\right\}
_{i=1,...,I}$ on the ON basis $U$ so that $B_{i}$ is a basis for the
eigenspace of the eigenvalue $\phi_{i}$ and $\left\vert B_{i}\right\vert $ is
the "multiplicity" (dimension of the eigenspace) of the eigenvalue $\phi_{i}$
for $i=1,...,I$. Note that the real-valued function $f:U\rightarrow%
\mathbb{R}
$ that takes each eigenvector in $u_{j}\in B_{i}\subseteq U$ to its eigenvalue
$\phi_{i}$ so that $f^{-1}\left(  \phi_{i}\right)  =B_{i}$ contains all the
information in the self-adjoint operator $F:V\rightarrow V$ since $F$ can be
reconstructed by defining it on the basis $U$ as $Fu_{j}=f\left(
u_{j}\right)  u_{j}$.

The generalization of `classical' logical entropy to quantum logical entropy
is straightforward using the usual ways that set-concepts generalize to
vector-space concepts: subsets $\rightarrow$ subspaces, set partitions
$\rightarrow$ direct-sum decompositions of subspaces\footnote{Hence the
`classical' logic of partitions on a set will generalize to the quantum logic
of direct-sum decompositions that is the dual to the usual quantum logic of
subspaces \cite{ell:qpl}.}, Cartesian products of sets $\rightarrow$ tensor
products of vector spaces, and ordered pairs $\left(  u_{k},u_{k^{\prime}%
}\right)  \in U\times U\rightarrow$ basis elements $u_{k}\otimes u_{k^{\prime
}}\in V\otimes V$. The eigenvalue function $f:U\rightarrow%
\mathbb{R}
$ determines a partition $\left\{  f^{-1}\left(  \phi_{i}\right)  \right\}
_{i\in I}$ on $U$ and the blocks in that partition generate the eigenspaces of
$F$ which form a direct-sum decomposition of $V$.

Classically, a \textit{dit of the partition} $\left\{  f^{-1}\left(  \phi
_{i}\right)  \right\}  _{i\in I}$ on $U$ is a pair $\left(  u_{k}%
,u_{k^{\prime}}\right)  $ of points in distinct blocks of the partition, i.e.,
$f\left(  u_{k}\right)  \neq f\left(  u_{k^{\prime}}\right)  $. Hence a
\textit{qudit of }$F$ is a pair $\left(  u_{k},u_{k^{\prime}}\right)  $
(interpreted as $u_{k}\otimes u_{k^{\prime}}$ in the context of $V\otimes V$)
of vectors in the eigenbasis definitely distinguishable by $F$, i.e.,
$f\left(  u_{k}\right)  \neq f\left(  u_{k^{\prime}}\right)  $, distinct
$F$-eigenvalues. Let $G:V\rightarrow V$ be another self-adjoint operator on
$V$ which commutes with $F$ so that we may then assume that $U$ is an
orthonormal basis of simultaneous eigenvectors of $F$ and $G$. Let $\left\{
\gamma_{j}\right\}  _{j\in J}$ be the set of eigenvalues of $G$ and let
$g:U\rightarrow%
\mathbb{R}
$ be the eigenvalue function so a pair $\left(  u_{k},u_{k^{\prime}}\right)  $
is a \textit{qudit of }$G$ if $g\left(  u_{k}\right)  \neq g\left(
u_{k^{\prime}}\right)  $, i.e., if the two eigenvectors have distinct
eigenvalues of $G$.

As in classical logical information theory, information is represented by
certain subsets--or, in the quantum case, subspaces--prior to the introduction
of any probabilities. Since the transition from classical to quantum logical
information theory is straightforward, it will be presented in table form in
Table 3a (which does not involve any probabilities)--where the qudits $\left(
u_{k},u_{k^{\prime}}\right)  $ are interpreted as $u_{k}\otimes u_{k^{\prime}%
}$.

\begin{center}%
\begin{center}
\includegraphics[
height=2.3638in,
width=4.1033in
]%
{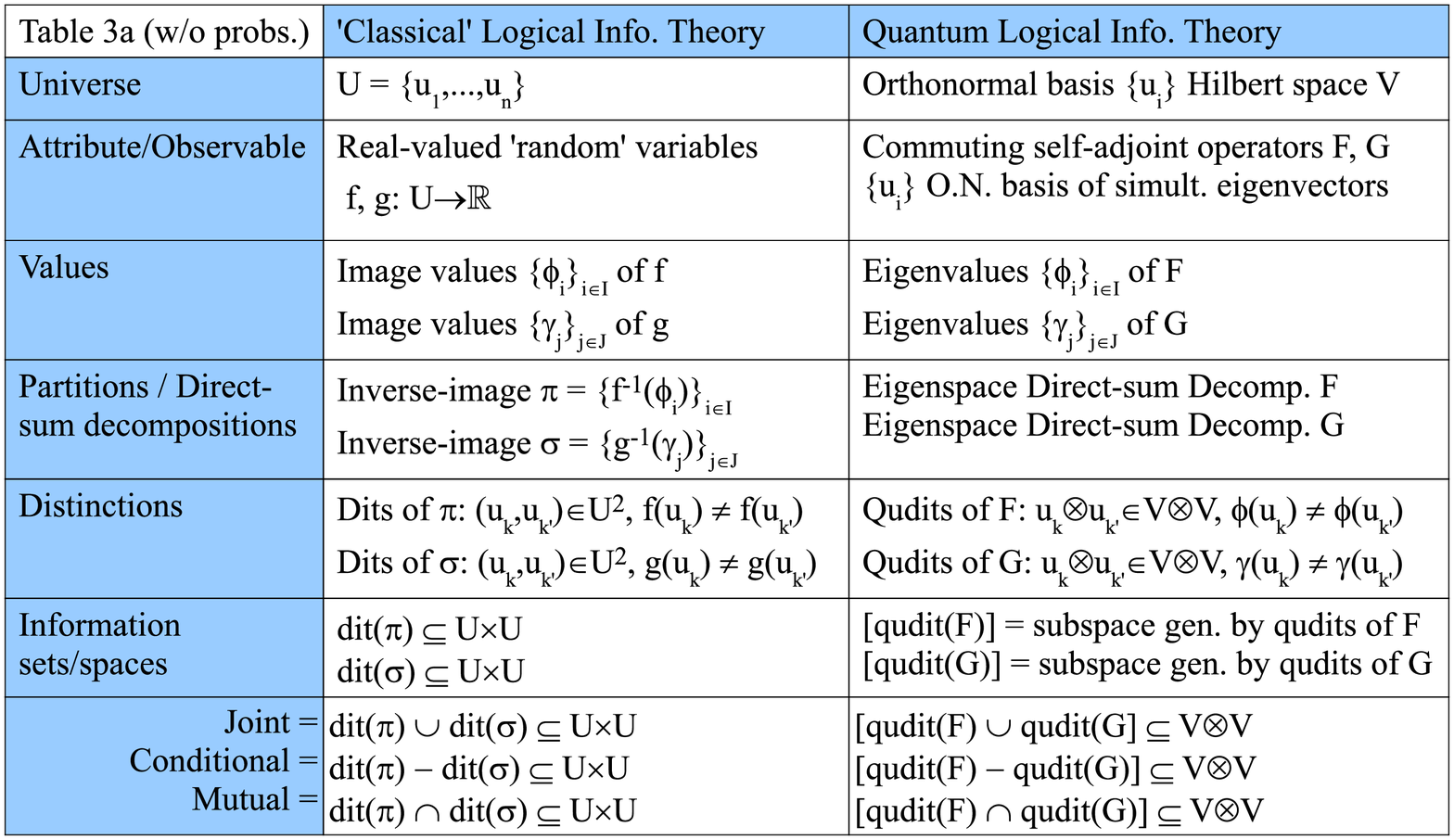}%
\end{center}
Table 3a: The parallel development of classical and quantum logical information

prior to probabilities.
\end{center}

\noindent The \textit{information subspace} associated with $F$ is the
subspace $\left[  qudit\left(  F\right)  \right]  \subseteq V\otimes V$
generated by the qudits $u_{k}\otimes u_{k^{\prime}}$ of $F$. If $F=\lambda I$
is a scalar multiple of the identity $I$, then it has no qudits so its
information space $\left[  qudit\left(  \lambda I\right)  \right]  $ is the
zero subspace. It is an easy implication of the Common Dits Theorem of
classical logical information theory (\cite[Proposition 1]{ell:countingdits}
or \cite[Theorem 1.4]{ell:partitions}) that any two nonzero information spaces
$\left[  qudit\left(  F\right)  \right]  $ and $\left[  qudit\left(  G\right)
\right]  $ have a nonzero intersection, i.e., have a nonzero mutual
information space. That is, there are always two eigenvectors $u_{k}$ and
$u_{k^{\prime}}$ that have different eigenvalues by both $F$ and $G$.

In a measurement, the observables do not provide the point probabilities; they
come from the pure (normalized) state $\psi$ being measured. Let $\left\vert
\psi\right\rangle =\sum_{j=1}^{n}\left\langle u_{j}|\psi\right\rangle
\left\vert u_{j}\right\rangle =\sum_{j=1}^{n}\alpha_{j}\left\vert
u_{j}\right\rangle $ be the resolution of $\left\vert \psi\right\rangle $ in
terms of the orthonormal basis $U=\left\{  u_{1},...,u_{n}\right\}  $ of
simultaneous eigenvectors for $F$ and $G$. Then $p_{j}=\alpha_{j}\alpha
_{j}^{\ast}$ ($\alpha_{j}^{\ast}$ is the complex conjugate of $\alpha_{j}$)
for $j=1,...,n$ are the point probabilities on $U$ and the pure state density
matrix $\rho\left(  \psi\right)  =\left\vert \psi\right\rangle \left\langle
\psi\right\vert $ (where $\left\langle \psi\right\vert =\left\vert
\psi\right\rangle ^{\dagger}$ is the conjugate-transpose) has the entries:
$\rho_{jk}\left(  \psi\right)  =\alpha_{j}\alpha_{k}^{\ast}$ so the diagonal
entries $\rho_{jj}\left(  \psi\right)  =\alpha_{j}\alpha_{j}^{\ast}=p_{j}$ are
the point probabilities. Table 3b gives the remaining parallel development
with the probabilities provided by the pure state $\psi$ where we write
$\rho\left(  \psi\right)  ^{\dagger}\rho\left(  \psi\right)  $ as $\rho\left(
\psi\right)  ^{2}$.%

\begin{center}
\includegraphics[
height=2.3828in,
width=3.8521in
]%
{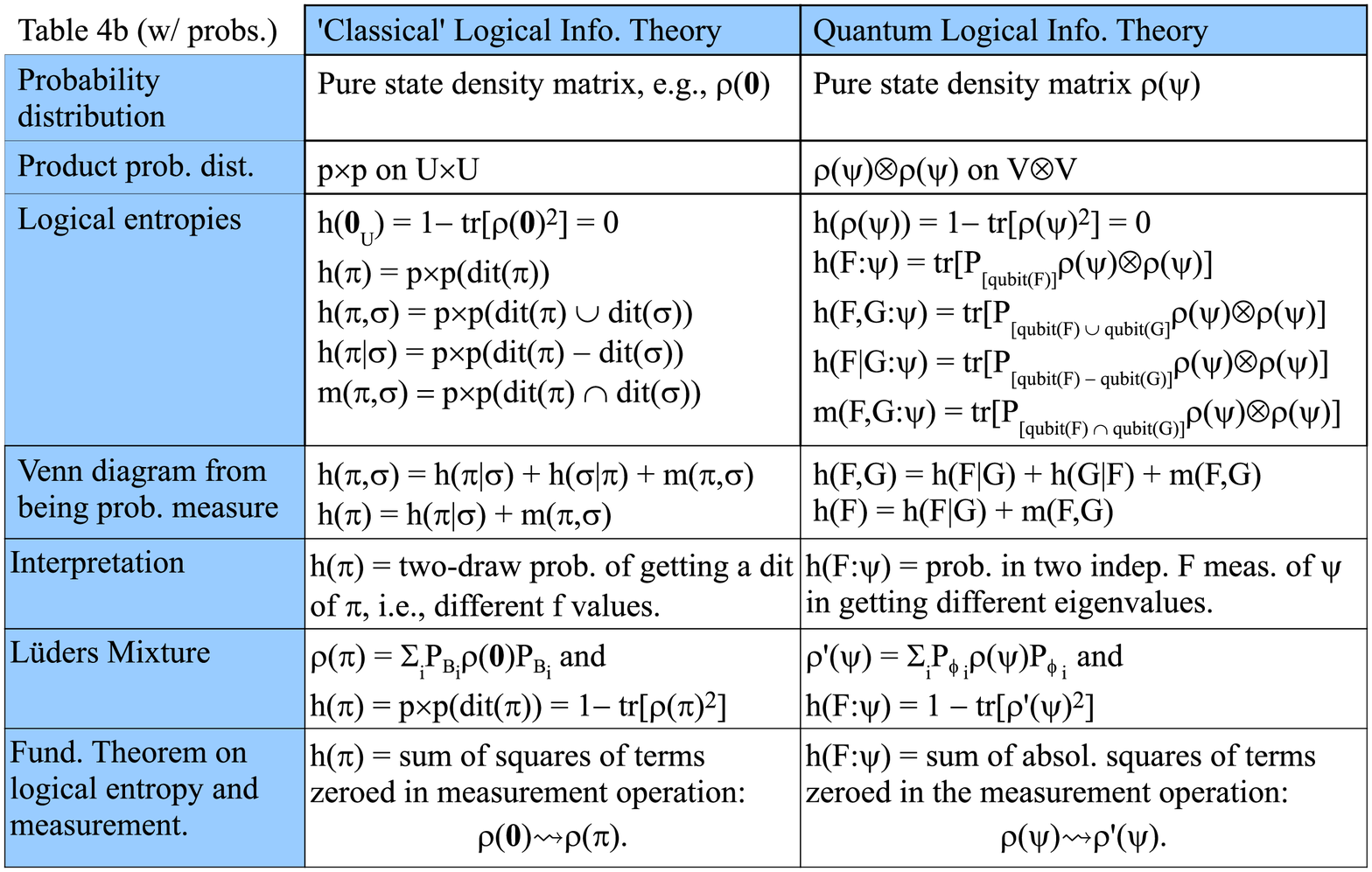}%
\end{center}

\begin{center}
Table 3b: The parallel development of classical and quantum logical entropies

for commuting $F$ and $G$.
\end{center}

The formula $h\left(  \rho\right)  =1-\operatorname*{tr}\left[  \rho
^{2}\right]  $ is hardly new. Indeed, $\operatorname*{tr}\left[  \rho
^{2}\right]  $ is usually called the \textit{purity} of the density matrix
since a state $\rho$ is \textit{pure} if and only if $\operatorname*{tr}%
\left[  \rho^{2}\right]  =1$ so $h\left(  \rho\right)  =0$, and otherwise
$\operatorname*{tr}\left[  \rho^{2}\right]  \,<1$ so $h\left(  \rho\right)
>0$ and the state is said to be \textit{mixed}. Hence the complement
$1-\operatorname*{tr}\left[  \rho^{2}\right]  $ has been called the
\textquotedblleft mixedness\textquotedblright\ \cite[p. 5]{jaeger:qinfo} or
\textquotedblleft impurity\textquotedblright\ of the state $\rho
$.\footnote{\noindent It is also called by the misnomer \textquotedblleft
linear entropy\textquotedblright\ \cite{buscemi:lin-entropy} even though it is
obviously quadratic in $\rho$--so we will not continue that usage. The logical
entropy is also the quadratic special case of the Tsallis-Havrda-Charvat
entropy (\cite{havrda:alpha}, \cite{tsallis:entropy}), and the logical special
case \cite{ell:countingdits} of C. R. Rao's quadratic entropy \cite{rao:div}.}
What is new is not the formula but the whole backstory of partition logic
outlined above which gives the logical notion of entropy arising out of
partition logic as the normalized counting measure on ditsets--just as logical
probability arises out of Boolean subset logic as the normalized counting
measure on subsets. The basic idea of information is differences,
distinguishability, and distinctions (\cite{ell:countingdits}, \cite{ell:lit}%
), so the logical notion of entropy is the measure of the distinctions or dits
of a partition and the corresponding quantum version is the measure of the
qudits of an observable.

The classical dit-bit transform connecting the logical theory to the Shannon
theory also carries over to the quantum version. Writing the quantum logical
entropy of a density matrix $\rho$ as $h\left(  \rho\right)
=\operatorname*{tr}\left[  \rho\left(  1-\rho\right)  \right]  $, the quantum
version of the dit-bit transform $\left(  1-\rho\right)  \leadsto-\log\left(
\rho\right)  $ yields the usual Von Neumann entropy $S\left(  \rho\right)
=-\operatorname*{tr}\left[  \rho\log\left(  \rho\right)  \right]  $\cite[p.
510]{nielsen-chuang:bible}. The fundamental theorem connecting logical entropy
and the operation of classification-measurement also carries over to the
quantum case.

\section{Two Theorems about quantum logical entropy}

Classically, a pair of elements $\left(  u_{j},u_{k}\right)  $ either `cohere'
together in the same block of a partition on $U$, i.e., are an indistinction
of the partition, or they don't, i.e., are a distinction of the partition. In
the quantum case, the nonzero off-diagonal entries $\alpha_{j}\alpha_{k}%
^{\ast}$ in the pure state density matrix $\rho\left(  \psi\right)
=\left\vert \psi\right\rangle \left\langle \psi\right\vert $ are called
quantum \textquotedblleft coherences\textquotedblright\ (\cite[p.
303]{cohen-t:QM1}; \cite[p.177]{auletta:qm}) because they give the amplitude
of the eigenstates $\left\vert u_{j}\right\rangle $ and $\left\vert
u_{k}\right\rangle $ \textquotedblleft cohering\textquotedblright\ together in
the coherent superposition state vector $\left\vert \psi\right\rangle
=\sum_{j=1}^{n}\left\langle u_{j}|\psi\right\rangle \left\vert u_{j}%
\right\rangle =\sum_{j}\alpha_{j}\left\vert u_{j}\right\rangle $. The
coherences are classically modelled by the nonzero off-diagonal entries
$\sqrt{p_{j}p_{k}}$ for the indistinctions $\left(  u_{j},u_{k}\right)  \in
B_{i}\times B_{i}$, i.e., coherences $\approx$ indistinctions.

For an observable $F$, let $\phi:U\rightarrow%
\mathbb{R}
$ be for $F$-eigenvalue function assigning the eigenvalue $\phi\left(
u_{j}\right)  =\phi_{i}$ for each $u_{j}$ in the ON basis $U=\left\{
u_{1},...,u_{n}\right\}  $ of $F$-eigenvectors. The range of $\phi$ is the set
of $F$-eigenvalues $\left\{  \phi_{1},...,\phi_{I}\right\}  $. Let
$P_{\phi_{i}}:V\rightarrow V$ be the projection matrix in the $U$-basis to the
eigenspace of $\phi_{i}$. The projective $F$-measurement of the state $\psi$
transforms the pure state density matrix $\rho\left(  \psi\right)  $ to yield
the L\"{u}ders mixture density matrix $\rho^{\prime}\left(  \psi\right)
=\sum_{i=1}^{I}P_{\phi_{i}}\rho\left(  \psi\right)  P_{\phi_{i}}$ \cite[p.
279]{auletta:qm}. The off-diagonal elements of $\rho\left(  \psi\right)  $
that are zeroed in $\rho^{\prime}\left(  \psi\right)  $ are the coherences
(quantum indistinctions or \textit{quindits}) that are turned into
`decoherences' (quantum distinctions or qudits of the observable being measured).

For any observable $F$ and a pure state $\psi$, a quantum logical entropy was
defined as $h\left(  F:\psi\right)  =\operatorname*{tr}\left[  P_{\left[
qudit\left(  F\right)  \right]  }\rho\left(  \psi\right)  \otimes\rho\left(
\psi\right)  \right]  $. That definition was the quantum generalization of the
`classical' logical entropy defined as $h\left(  \pi\right)  =p\times p\left(
\operatorname*{dit}\left(  \pi\right)  \right)  $. When a projective
$F$-measurement is performed on $\psi$, the pure state density matrix
$\rho\left(  \psi\right)  $ is transformed into the mixed state density matrix
by the quantum L\"{u}ders mixture operation which then defines the quantum
logical entropy $h\left(  \rho^{\prime}\left(  \psi\right)  \right)
=1-\operatorname*{tr}\left[  \rho^{\prime}\left(  \psi\right)  ^{2}\right]  $.
The first test of how the quantum logical entropy notions fit together is
showing that these two entropies are the same: $h\left(  F:\psi\right)
=h\left(  \rho^{\prime}\left(  \psi\right)  \right)  $. The proof shows that
they are both equal to classical logical entropy of the partition $\pi\left(
F:\psi\right)  $ defined on the ON basis $U=\left\{  u_{1},...,u_{n}\right\}
$ of $F$-eigenvectors by the $F$-eigenvalues with the point probabilties
$p_{j}=\alpha_{j}^{\ast}\alpha_{j}$. That is, the inverse-images $B_{i}%
=\phi^{-1}\left(  \phi_{i}\right)  $ of the eigenvalue function $\phi
:U\rightarrow%
\mathbb{R}
$ define the eigenvalue partition $\pi\left(  F:\psi\right)  =\left\{
B_{1},...,B_{I}\right\}  $ on the ON basis $U=\left\{  u_{1},...,u_{n}%
\right\}  $ with the point probabilities $p_{j}=\alpha_{j}^{\ast}\alpha_{j}$
provided by the state $\left\vert \psi\right\rangle $ for $j=1,...,n$. The
classical logical entropy of that partition is: $h\left(  \pi\left(
F:\psi\right)  \right)  =1-\sum_{i=1}^{I}p\left(  B_{i}\right)  ^{2}$ where
$p\left(  B_{i}\right)  =\sum_{u_{j}\in B_{i}}p_{j}$.

We first show that $h\left(  F:\psi\right)  =\operatorname*{tr}\left[
P_{\left[  qudit\left(  F\right)  \right]  }\rho\left(  \psi\right)
\otimes\rho\left(  \psi\right)  \right]  =h\left(  \pi\left(  F:\psi\right)
\right)  $. Now $qudit\left(  F\right)  =\left\{  u_{j}\otimes u_{k}%
:\phi\left(  u_{j}\right)  \neq\phi\left(  u_{k}\right)  \right\}  $ and
$\left[  qudit\left(  F\right)  \right]  $ is the subspace of $V\otimes V$
generated by it. The $n\times n$ pure state density matrix $\rho\left(
\psi\right)  $ has the entries $\rho_{jk}\left(  \psi\right)  =\alpha
_{j}\alpha_{k}^{\ast}$, and $\rho\left(  \psi\right)  \otimes\rho\left(
\psi\right)  $ is an $n^{2}\times n^{2}$ matrix. The projection matrix
$P_{\left[  qudit\left(  F\right)  \right]  }$ is an $n^{2}\times n^{2}$
diagonal matrix with the diagonal entries, indexed by $j,k=1,...,n$: $\left[
P_{\left[  qudit\left(  F\right)  \right]  }\right]  _{jjkk}=1$ if
$\phi\left(  u_{j}\right)  \neq\phi\left(  u_{k}\right)  $ and $0$ otherwise.
Thus in the product $P_{\left[  qudit(F)\right]  }\rho\left(  \psi\right)
\otimes\rho\left(  \psi\right)  $, the nonzero diagonal elements are the
$p_{j}p_{k}$ where $\phi\left(  u_{j}\right)  \neq\phi\left(  u_{k}\right)  $
and so the trace is $\sum_{j.k=1}^{n}\left\{  p_{j}p_{k}:\phi\left(
u_{j}\right)  \neq\phi\left(  u_{k}\right)  \right\}  $ which, by definition,
is $h\left(  F:\psi\right)  $. Since $\sum_{j=1}^{n}p_{j}=\sum_{i=1}%
^{I}p\left(  B_{i}\right)  =1$, $\left(  \sum_{i=1}^{I}p\left(  B_{i}\right)
\right)  ^{2}=1=\sum_{i=1}^{I}p\left(  B_{i}\right)  ^{2}+\sum_{i\neq
i^{\prime}}p\left(  B_{i}\right)  p\left(  B_{i^{\prime}}\right)  $. By
grouping the $p_{j}p_{k}$ in the trace according to the blocks of $\pi\left(
F:\psi\right)  $, we have:

\begin{center}
$h\left(  F:\psi\right)  =\operatorname*{tr}\left[  P_{\left[
qudit(F)\right]  }\rho\left(  \psi\right)  \otimes\rho\left(  \psi\right)
\right]  =\sum_{j.k=1}^{n}\left\{  p_{j}p_{k}:\phi\left(  u_{j}\right)
\neq\phi\left(  u_{k}\right)  \right\}  $

$=\sum_{i\neq i^{\prime}}\sum\left\{  p_{j}p_{k}:u_{j}\in B_{i},u_{k}\in
B_{i^{\prime}}\right\}  =\sum_{i\neq i^{\prime}}p\left(  B_{i}\right)
p\left(  B_{i^{\prime}}\right)  $

$=1-\sum_{i=1}^{I}p\left(  B_{i}\right)  ^{2}=h\left(  \pi\left(
F:\psi\right)  \right)  $.
\end{center}

To show that $h\left(  \rho^{\prime}\left(  \psi\right)  \right)
=1-\operatorname*{tr}\left[  \rho^{\prime}\left(  \psi\right)  ^{2}\right]
=h\left(  \pi\left(  F:\psi\right)  \right)  $ for $\rho^{\prime}\left(
\psi\right)  =\sum_{i=1}^{I}P_{\phi_{i}}\rho\left(  \psi\right)  P_{\phi_{i}}%
$, we need to compute $\operatorname*{tr}\left[  \rho^{\prime}\left(
\psi\right)  ^{2}\right]  $. An off-diagonal element in $\rho_{jk}\left(
\psi\right)  =\alpha_{j}\alpha_{k}^{\ast}$ of $\rho\left(  \psi\right)  $
survives (i.e., is not zeroed and has the same value) the L\"{u}ders operation
if and only if $\phi\left(  u_{j}\right)  =\phi\left(  u_{k}\right)  $. Hence
the $j^{th}$ diagonal element of $\rho^{\prime}\left(  \psi\right)  ^{2}$ is

\begin{center}
$\sum_{k=1}^{n}\left\{  \alpha_{j}^{\ast}\alpha_{k}\alpha_{j}\alpha_{k}^{\ast
}:\phi\left(  u_{j}\right)  =\phi\left(  u_{k}\right)  \right\}  =\sum
_{k=1}^{n}\left\{  p_{j}p_{k}:\phi\left(  u_{j}\right)  =\phi\left(
u_{k}\right)  \right\}  =p_{j}p\left(  B_{i}\right)  $
\end{center}

\noindent where $u_{j}\in B_{i}$. Then grouping the $j^{th}$ diagonal elements
for $u_{j}\in B_{i}$ gives $\sum_{u_{j}\in B_{i}}p_{j}p\left(  B_{i}\right)
=p\left(  B_{i}\right)  ^{2}$. Hence the whole trace is: $\operatorname*{tr}%
\left[  \rho^{\prime}\left(  \psi\right)  ^{2}\right]  =\sum_{i=1}^{I}p\left(
B_{i}\right)  ^{2}$ and thus:

\begin{center}
$h\left(  \rho^{\prime}\left(  \psi\right)  \right)  =1-\operatorname*{tr}%
\left[  \rho^{\prime}\left(  \psi\right)  ^{2}\right]  =1-\sum_{i=1}%
^{I}p\left(  B_{i}\right)  ^{2}=h\left(  F:\psi\right)  $.
\end{center}

\noindent This completes the proof of the following theorem.

\begin{theorem}
$h\left(  F:\psi\right)  =h\left(  \pi\left(  F:\psi\right)  \right)
=h\left(  \rho^{\prime}\left(  \psi\right)  \right)  $.$\blacksquare$
\end{theorem}

Measurement creates distinctions, i.e., turns coherences into
`decoherences'--which, classically, is the operation of distinguishing
elements by classifying them according to some attribute like classifying the
faces of a die by their parity. The fundamental theorem about quantum logical
entropy and projective measurement shows how the quantum logical entropy
created (starting with $h\left(  \rho\left(  \psi\right)  \right)  =0$ for the
pure state $\psi$) by the measurement can be computed directly from the
coherences of $\rho\left(  \psi\right)  $ that are decohered in $\rho^{\prime
}\left(  \psi\right)  $.

\begin{theorem}
[Fundamental]The increase in quantum logical entropy, $h\left(  F:\psi\right)
=h\left(  \rho^{\prime}\left(  \psi\right)  \right)  $, due to the
$F$-measurement of the pure state $\psi$ is the sum of the absolute squares of
the nonzero off-diagonal terms [coherences] in $\rho\left(  \psi\right)  $
that are zeroed [`decohered'] in the post-measurement mixed state density
matrix $\rho^{\prime}\left(  \psi\right)  =\sum_{i=1}^{I}P_{\phi_{i}}%
\rho\left(  \psi\right)  P_{\phi_{i}}$.
\end{theorem}

\begin{proof}
$h\left(  \rho^{\prime}\left(  \psi\right)  \right)  -h\left(  \rho\left(
\psi\right)  \right)  =\left(  1-\operatorname*{tr}\left[  \rho^{\prime
}\left(  \psi\right)  ^{2}\right]  \right)  -\left(  1-\operatorname*{tr}%
\left[  \rho\left(  \psi\right)  ^{2}\right]  \right)  =\sum_{jk}\left(
\left\vert \rho_{jk}\left(  \psi\right)  \right\vert ^{2}-\left\vert \rho
_{jk}^{\prime}\left(  \psi\right)  \right\vert ^{2}\right)  $. If $u_{j}$ and
$u_{k}$ are a qudit of $F$, then and only then are the corresponding
off-diagonal terms zeroed by the L\"{u}ders mixture operation $\sum_{i=1}%
^{I}P_{\phi_{i}}\rho\left(  \psi\right)  P_{\phi_{i}}$ to obtain $\rho
^{\prime}\left(  \psi\right)  $ from $\rho\left(  \psi\right)  $.
\end{proof}

Density matrices have long been a standard part of the machinery of quantum
mechanics. The Fundamental Theorem for logical entropy and measurement shows
there is a simple, direct, and quantitative connection between density
matrices and logical entropy. The Theorem directly connects the changes in the
density matrix due to a measurement (sum of absolute squares of zeroed
off-diagonal terms) with the increase in logical entropy due to the
$F$-measurement $h\left(  F:\psi\right)  =\operatorname*{tr}\left[  P_{\left[
qudit(F)\right]  }\rho\left(  \psi\right)  \otimes\rho\left(  \psi\right)
\right]  =h\left(  \rho^{\prime}\left(  \psi\right)  \right)  $ (where
$h\left(  \rho\left(  \psi\right)  \right)  =0$ for the pure state $\psi$).

This direct quantitative connection between state discrimination and quantum
logical entropy reinforces the judgment of Boaz Tamir and Eli Cohen
(\cite{tamir-cohen:logicalentropy}, \cite{tamir-cohen:hilbert-schmidt}) that
quantum logical entropy is a natural and informative entropy concept for
quantum mechanics.

\begin{quotation}
\noindent We find this framework of partitions and distinction most suitable
(at least conceptually) for describing the problems of quantum state
discrimination, quantum cryptography and in general, for discussing quantum
channel capacity. In these problems, we are basically interested in a distance
measure between such sets of states, and this is exactly the kind of knowledge
provided by logical entropy [Reference to \cite{ell:countingdits}]. \cite[p.
1]{tamir-cohen:logicalentropy}
\end{quotation}

\noindent Moreover, the quantum logical entropy has a simple \textquotedblleft
two-draw probability\textquotedblright\ interpretation, i.e., $h\left(
F:\psi\right)  =h\left(  \rho^{\prime}\left(  \psi\right)  \right)  $ is the
probability that two independent $F$-measurements of $\psi$ will yield
distinct $F$-eigenvalues, i.e., will yield a qudit of $F$. In contrast, the
Von Neumann entropy has no such simple interpretation and there seems to be no
such intuitive connection between pre- and post-measurement density matrices
and Von Neumann entropy--although Von Neumann entropy also increases in a
projective measurement \cite[Theorem 11.9, p. 515]{nielsen-chuang:bible}.

\section{Quantum Logical Information Theory: Non-commuting Observables}

\subsection{Classical\ logical information theory with two sets $X$ and $Y$}

The usual (`classical') logical information theory for a probability
distribution $\left\{  p\left(  x,y\right)  \right\}  $ on $X\times Y$
(finite) in effect uses the discrete partition on $X$ and $Y$ \cite{ell:lit}.
For the general case of quantum logical entropy for not-necessarily commuting
observables, we need to first briefly develop the classical case with general
partitions on $X$ and $Y$.

Given two finite sets $X$ and $Y$ and real-valued functions $f:X\rightarrow%
\mathbb{R}
$ with values $\left\{  \phi_{i}\right\}  _{i=1}^{I}$ and $g:Y\rightarrow%
\mathbb{R}
$ with values $\left\{  \gamma_{j}\right\}  _{j=1}^{J}$, each function induces
a partition on its domain:

\begin{center}
$\pi=\left\{  f^{-1}\left(  \phi_{i}\right)  \right\}  _{i\in I}=\left\{
B_{1},...,B_{I}\right\}  $ on $X$, and $\sigma=\left\{  g^{-1}\left(
\gamma_{j}\right)  \right\}  _{j\in J}=\left\{  C_{1},...,C_{J}\right\}  $ on
$Y$.
\end{center}

We need to define logical entropies on $X\times Y$ but first we need to define
the ditsets or information sets.

A partition $\pi=\left\{  B_{1},...,B_{I}\right\}  $ on $X$ and a partition
$\sigma=\left\{  C_{1},...,C_{J}\right\}  $ on $Y$ define a \textit{product
partition} $\pi\times\sigma$ on $X\times Y$ whose blocks are $\left\{
B_{i}\times C_{j}\right\}  _{i,j}$. Then $\pi$ induces $\pi\times
\mathbf{0}_{Y}$ on $X\times Y$ (where $\mathbf{0}_{Y}$ is the indiscrete
partition on $Y$) and $\sigma$ induces $\mathbf{0}_{X}\times\sigma$ on
$X\times Y$. The corresponding ditsets or information sets are:

\begin{itemize}
\item $\operatorname*{dit}\left(  \pi\times\mathbf{0}_{Y}\right)  =\left\{
\left(  \left(  x,y\right)  ,\left(  x^{\prime},y^{\prime}\right)  \right)
:f\left(  x\right)  \neq f\left(  x^{\prime}\right)  \right\}  \subseteq
\left(  X\times Y\right)  ^{2}$;

\item $\operatorname*{dit}\left(  \mathbf{0}_{X}\times\sigma\right)  =\left\{
\left(  \left(  x,y\right)  ,\left(  x^{\prime},y^{\prime}\right)  \right)
:g\left(  y\right)  \neq g\left(  y^{\prime}\right)  \right\}  \subseteq
\left(  X\times Y\right)  ^{2}$;

\item $\operatorname*{dit}\left(  \pi\times\sigma\right)  =\operatorname*{dit}%
\left(  \pi\times\mathbf{0}_{Y}\right)  \cup\operatorname*{dit}\left(
\mathbf{0}_{X}\times\sigma\right)  $; and so forth.
\end{itemize}

Given a joint probability distribution $p:X\times Y\rightarrow\left[
0,1\right]  $, the product probability distribution is $p\times p:\left(
X\times Y\right)  ^{2}\rightarrow\left[  0,1\right]  $.

All the logical entropies are just the product probabilities of the ditsets
and their union, differences, and intersection:

\begin{itemize}
\item $h\left(  \pi\times\mathbf{0}_{Y}\right)  =p\times p\left(
\operatorname*{dit}\left(  \pi\times\mathbf{0}_{Y}\right)  \right)  $;

\item $h\left(  \mathbf{0}_{X}\times\sigma\right)  =p\times p\left(
\operatorname*{dit}\left(  \mathbf{0}_{X}\times\sigma\right)  \right)  $;

\item $h\left(  \pi\times\sigma\right)  =p\times p\left(  \operatorname*{dit}%
\left(  \pi\times\sigma\right)  \right)  =p\times p\left(  \operatorname*{dit}%
\left(  \pi\times\mathbf{0}_{Y}\right)  \cup\operatorname*{dit}\left(
\mathbf{0}_{X}\times\sigma\right)  \right)  $;

\item $h\left(  \pi\times\mathbf{0}_{Y}|\mathbf{0}_{X}\times\sigma\right)
=p\times p\left(  \operatorname*{dit}\left(  \pi\times\mathbf{0}_{Y}\right)
-\operatorname*{dit}\left(  \mathbf{0}_{X}\times\sigma\right)  \right)  $;

\item $h\left(  \mathbf{0}_{X}\times\sigma|\pi\times\mathbf{0}_{Y}\right)
=p\times p\left(  \operatorname*{dit}\left(  \mathbf{0}_{X}\times
\sigma\right)  -\operatorname*{dit}\left(  \pi\times\mathbf{0}_{Y}\right)
\right)  $;

\item $m\left(  \pi\times\mathbf{0}_{Y},\mathbf{0}_{X}\times\sigma\right)
=p\times p\left(  \operatorname*{dit}\left(  \pi\times\mathbf{0}_{Y}\right)
\cap\operatorname*{dit}\left(  \mathbf{0}_{X}\times\sigma\right)  \right)  $.
\end{itemize}

All the logical entropies have the usual two-draw probability interpretation
where the two independent draws from $X\times Y$ are $\left(  x,y\right)  $
and $\left(  x^{\prime},y^{\prime}\right)  $ and can be interpreted in terms
of the $f$-values and $g$-values:

\begin{itemize}
\item $h\left(  \pi\times\mathbf{0}_{Y}\right)  $ = probability of getting
distinct $f$-values;

\item $h\left(  \mathbf{0}_{X}\times\sigma\right)  $ = probability of getting
distinct $g$-values;

\item $h\left(  \pi\times\sigma\right)  $ = probability of getting distinct
$f$ or $g$ values;

\item $h\left(  \pi\times\mathbf{0}_{Y}|\mathbf{0}_{X}\times\sigma\right)  $ =
probability of getting distinct $f$-values but same $g$-values;

\item $h\left(  \mathbf{0}_{X}\times\sigma|\pi\times\mathbf{0}_{Y}\right)  $ =
probability of getting distinct $g$-values but same $f$-values;

\item $m\left(  \pi\times\mathbf{0}_{Y},\mathbf{0}_{X}\times\sigma\right)  $ =
probability of getting distinct $f$ and $g$ values.
\end{itemize}

We have defined all the logical entropies by the general method of the product
probabilities on the ditsets. In the first three cases, $h\left(  \pi
\times\mathbf{0}_{Y}\right)  $, $h\left(  \mathbf{0}_{X}\times\sigma\right)
$, and $h\left(  \pi\times\sigma\right)  $, they were the logical entropies of
partitions on $X\times Y$ so they could equivalently be defined using density
matrices. The case of $h\left(  \pi\times\sigma\right)  $ illustrates the
general case. If $\rho\left(  \pi\right)  $ is the density matrix defined for
$\pi$ on $X$ and $\rho\left(  \sigma\right)  $ the density matrix for $\sigma$
on $Y$, then $\rho\left(  \pi\times\sigma\right)  =\rho\left(  \pi\right)
\otimes\rho\left(  \sigma\right)  $ is the density matrix for $\pi\times
\sigma$ defined on $X\times Y$, and:

\begin{center}
$h\left(  \pi\times\sigma\right)  =1-\operatorname*{tr}\left[  \rho\left(
\pi\times\sigma\right)  ^{2}\right]  $.
\end{center}

The marginal distributions: $p_{X}\left(  x\right)  =\sum_{y}p\left(
x,y\right)  $ and $p_{Y}\left(  y\right)  =\sum_{x}p\left(  x,y\right)  $.
Since $\pi$ is a partition on $X$, there is also the usual logical entropy
$h\left(  \pi\right)  =p_{X}\times p_{X}\left(  \operatorname*{dit}\left(
\pi\right)  \right)  =1-\operatorname*{tr}\left[  \rho\left(  \pi\right)
^{2}\right]  =h\left(  \pi\times\mathbf{0}_{Y}\right)  $ where
$\operatorname*{dit}\left(  \pi\right)  \subseteq X\times X$ and similarly for
$p_{Y}$.

Since the context should be clear, we may henceforth adopt the old notation
from the case where $\pi$ and $\sigma$ were partitions on the same set $U$,
i.e., $h\left(  \pi\right)  =h\left(  \pi\times\mathbf{0}_{Y}\right)  $,
$h\left(  \sigma\right)  =h\left(  \mathbf{0}_{X}\times\sigma\right)  $,
$h\left(  \pi,\sigma\right)  =h\left(  \pi\times\sigma\right)  $, etc.

Since the logical entropies are the values of a probability measure, all the
usual identities hold where the underlying set is now $\left(  X\times
Y\right)  ^{2}$ instead of $U^{2}$.%

\begin{center}
\includegraphics[
height=1.9796in,
width=2.335in
]%
{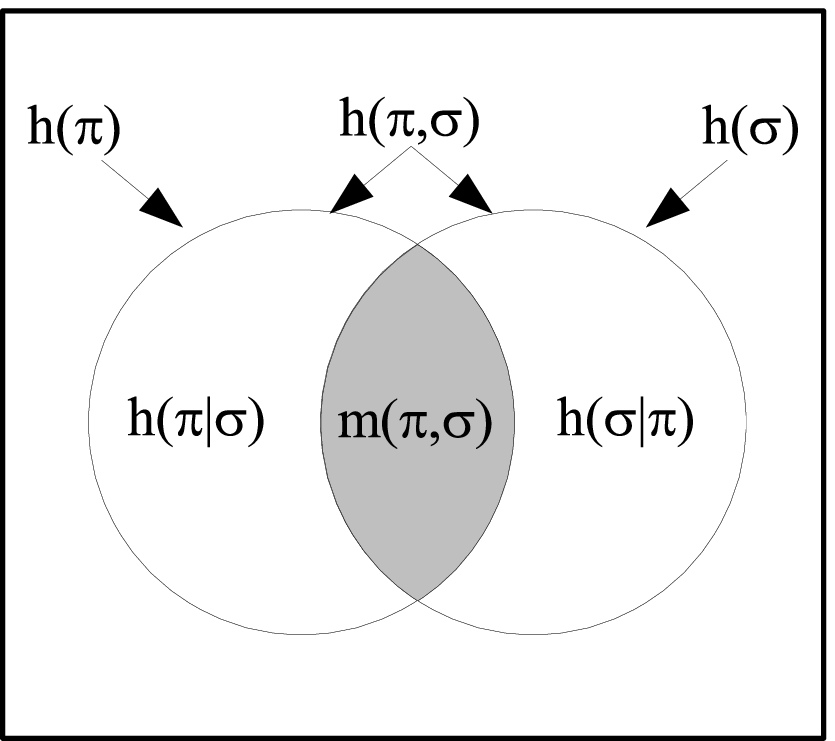}%
\end{center}

\begin{center}
Figure 4: Venn diagram for logical entropies

as values of a probability measure $p\times p$ on $\left(  X\times Y\right)
^{2}$.
\end{center}

The previous treatment of $h\left(  X\right)  $, $h\left(  Y\right)  $,
$h\left(  X,Y\right)  $, $h\left(  X|Y\right)  $, $h\left(  Y|X\right)  $, and
$m\left(  X,Y\right)  $ in \cite{ell:lit} was just the special cases where
$\pi=\mathbf{1}_{X}$ and $\sigma=\mathbf{1}_{Y}$.

\subsection{Quantum logical entropies with non-commuting observables}

As before in the case of commuting observables, the quantum case can be
developed in close analogy with the previous classical case. Given a
finite-dimensional Hilbert space $V$ and not necessarily commuting observables
$F,G:V\rightarrow V$, let $X$ be an orthonormal basis of $V$ of $F$%
-eigenvectors and let $Y$ be an orthonormal basis for $V$ of $G$-eigenvectors
(so $\left\vert X\right\vert =|Y|$).

Let $f:X\rightarrow%
\mathbb{R}
$ be the eigenvalue function for $F$ with values $\left\{  \phi_{i}\right\}
_{i=1}^{I}$, and let $g:Y\rightarrow%
\mathbb{R}
$ be the eigenvalue function for $G$ with values $\left\{  \gamma_{j}\right\}
_{j=1}^{J}$.

Each eigenvalue function induces a partition on its domain:

\begin{center}
$\pi=\left\{  f^{-1}\left(  \phi_{i}\right)  \right\}  =\left\{
B_{1},...,B_{I}\right\}  $ on $X$, and $\sigma=\left\{  g^{-1}\left(
\gamma_{j}\right)  \right\}  =\left\{  C_{1},...,C_{J}\right\}  $ on $Y$.
\end{center}

We associated with ordered pair $\left(  x,y\right)  $, the basis element
$x\otimes y$ in the basis $\left\{  x\otimes y\right\}  _{x\in X,y\in Y}$ for
$V\otimes V$. Then each pair of pairs $\left(  \left(  x,y\right)  ,\left(
x^{\prime},y^{\prime}\right)  \right)  $ is associated with the basis element
$\left(  x\otimes y\right)  \otimes\left(  x^{\prime}\otimes y^{\prime
}\right)  $ in $\left(  V\otimes V\right)  \otimes\left(  V\otimes V\right)
=\left(  V\otimes V\right)  ^{2}$.

Instead of ditsets or information sets, we now have qudit subspaces or
information subspaces. For $R\subseteq\left(  V\otimes V\right)  ^{2}$, let
$\left[  R\right]  $ be the subspace generated by $R$. We simplify notation of
$qudit\left(  \pi\times\mathbf{0}_{Y}\right)  =qudit\left(  \pi\right)
=\left\{  \left(  x\otimes y\right)  \otimes\left(  x^{\prime}\otimes
y^{\prime}\right)  :f\left(  x\right)  \neq f\left(  x^{\prime}\right)
\right\}  $, etc.

\begin{itemize}
\item $\left[  qudit\left(  \pi\right)  \right]  =\left[  \left\{  \left(
x\otimes y\right)  \otimes\left(  x^{\prime}\otimes y^{\prime}\right)
:f\left(  x\right)  \neq f\left(  x^{\prime}\right)  \right\}  \right]  $;

\item $\left[  qudit\left(  \sigma\right)  \right]  =\left[  \left\{  \left(
x\otimes y\right)  \otimes\left(  x^{\prime}\otimes y^{\prime}\right)
:g\left(  y\right)  \neq g\left(  y^{\prime}\right)  \right\}  \right]  $;

\item $\left[  qudit\left(  \pi,\sigma\right)  \right]  =\left[  qudit\left(
\pi\right)  \cup qudit\left(  \sigma\right)  \right]  $, and so
forth.\footnote{It is again an easy implication of the aforementioned Common
Dits Theorem that any two nonzero information spaces $\left[
\operatorname*{dit}\left(  \pi\right)  \right]  $ and $\left[
\operatorname*{dit}\left(  \sigma\right)  \right]  $ have a nonzero
intersection so the mutual information space $\left[  \operatorname*{dit}%
\left(  \pi\right)  \cap\operatorname*{dit}\left(  \sigma\right)  \right]  $
is not the zero space.}
\end{itemize}

A normalized state $\left\vert \psi\right\rangle $ on $V\otimes V$ defines a
pure state density matrix $\rho\left(  \psi\right)  =\left\vert \psi
\right\rangle \left\langle \psi\right\vert $. Let $\alpha_{x,y}=\left\langle
x\otimes y|\psi\right\rangle $ so if $P_{\left[  x\otimes y\right]  }$ is the
projection to the subspace (ray) generated by $x\otimes y$ in $V\otimes V$,
then a probability distribution on $X\times Y$ is defined by:

\begin{center}
$p\left(  x,y\right)  =\alpha_{x,y}\alpha_{x,y}^{\ast}=\operatorname*{tr}%
\left[  P_{\left[  x\otimes y\right]  }\rho\left(  \psi\right)  \right]  $,
\end{center}

\noindent or more generally, for a subspace $T\subseteq V\otimes V$, a
probability distribution is defined on the subspaces by:

\begin{center}
$\Pr\left(  T\right)  =\operatorname*{tr}\left[  P_{T}\rho\left(  \psi\right)
\right]  $.
\end{center}

Then the product probability distribution $p\times p$ on the subspaces of
$\left(  V\otimes V\right)  ^{2}$ defines the quantum logical entropies when
applied to the information subspaces:

\begin{itemize}
\item $h\left(  F:\psi\right)  =p\times p\left(  \left[  qudit\left(
\pi\right)  \right]  \right)  =\operatorname*{tr}\left[  P_{\left[
qudit\left(  \pi\right)  \right]  }\left(  \rho\left(  \psi\right)
\otimes\rho\left(  \psi\right)  \right)  \right]  $;

\item $h\left(  G:\psi\right)  =p\times p\left(  \left[  qudit\left(
\sigma\right)  \right]  \right)  =\operatorname*{tr}\left[  P_{\left[
qudit\left(  \sigma\right)  \right]  }\left(  \rho\left(  \psi\right)
\otimes\rho\left(  \psi\right)  \right)  \right]  $;

\item $h\left(  F,G:\psi\right)  =p\times p\left(  \left[  qudit\left(
\pi\right)  \cup qudit\left(  \sigma\right)  \right]  \right)
=\operatorname*{tr}\left[  P_{\left[  qudit\left(  \pi\right)  \cup
qudit\left(  \sigma\right)  \right]  }\left(  \rho\left(  \psi\right)
\otimes\rho\left(  \psi\right)  \right)  \right]  $;

\item $h\left(  F|G:\psi\right)  =p\times p\left(  \left[  qudit\left(
\pi\right)  -qudit\left(  \sigma\right)  \right]  \right)  =\operatorname*{tr}%
\left[  P_{\left[  qudit\left(  \pi\right)  -qudit\left(  \sigma\right)
\right]  }\left(  \rho\left(  \psi\right)  \otimes\rho\left(  \psi\right)
\right)  \right]  $;

\item $h\left(  G|F:\psi\right)  =p\times p\left(  \left[  qudit\left(
\sigma\right)  -qudit\left(  \pi\right)  \right]  \right)  =\operatorname*{tr}%
\left[  P_{\left[  qudit\left(  \sigma\right)  -qudit\left(  \pi\right)
\right]  }\left(  \rho\left(  \psi\right)  \otimes\rho\left(  \psi\right)
\right)  \right]  $;

\item $m\left(  F,G:\psi\right)  =p\times p\left(  \left[  qudit\left(
\pi\right)  \cap qudit\left(  \sigma\right)  \right]  \right)
=\operatorname*{tr}\left[  P_{\left[  qudit\left(  \pi\right)  \cap
qudit\left(  \sigma\right)  \right]  }\left(  \rho\left(  \psi\right)
\otimes\rho\left(  \psi\right)  \right)  \right]  $.
\end{itemize}

The observable $F:V\rightarrow V$ defines an observable $F\otimes I:V\otimes
V\rightarrow V\otimes V$ with the eigenvectors $x\otimes v$ for any nonzero
$v\in V$ and with the same eigenvalues $\phi_{1},...,\phi_{I}$.\footnote{The
context should suffice to distinguish the identity operator $I:V\rightarrow V$
from the index set $I$ for the $F$-eigenvalues.} Then in two independent
measurements of $\psi$ by the observable $F\otimes I$, we have:

\begin{center}
$h\left(  F:\psi\right)  $ = probability of getting distinct eigenvalues
$\phi_{i}$ and $\phi_{i^{\prime}}$, i.e., of getting a qudit of $F$.
\end{center}

In a similar manner, $G:V\rightarrow V$ defines the observable $I\otimes
G:V\otimes V\rightarrow V\otimes V$ with the eigenvectors $v\otimes y$ and
with the same eigenvalues $\gamma_{1},...,\gamma_{J}$. Then in two independent
measurements of $\psi$ by the observable $I\otimes G$, we have:

\begin{center}
$h\left(  G:\psi\right)  $ = probability of getting distinct eigenvalues
$\gamma_{j}$ and $\gamma_{j^{\prime}}$.
\end{center}

The two observables $F,G:V\rightarrow V$ define an observable $F\otimes
G:V\otimes V\rightarrow V\otimes V$ with the eigenvectors $x\otimes y$ for
$\left(  x,y\right)  \in X\times Y$ and eigenvalues $f\left(  x\right)
g\left(  y\right)  =\phi_{i}\gamma_{j}$. To cleanly interpret the compound
logical entropies, we assume there is no accidental degeneracy so there are no
$\phi_{i}\gamma_{j}=\phi_{i^{\prime}}\gamma_{j^{\prime}}$ for $i\neq
i^{\prime}$ and $j\neq j^{\prime}$. Then for two independent measurements of
$\psi$ by $F\otimes G$, the compound quantum logical entropies can be
interpreted as the following \textquotedblleft
two-measurement\textquotedblright\ probabilities:

\begin{itemize}
\item $h\left(  F,G:\psi\right)  $ = probability of getting distinct
eigenvalues $\phi_{i}\gamma_{j}\neq\phi_{i^{\prime}}\gamma_{j^{\prime}}$ where
$i\neq i^{\prime}$ or $j\neq j^{\prime}$;

\item $h\left(  F|G:\psi\right)  $ = probability of getting distinct
eigenvalues $\phi_{i}\gamma_{j}\neq\phi_{i^{\prime}}\gamma_{j}$ where $i\neq
i^{\prime}$;

\item $h\left(  G|F:\psi\right)  $ = probability of getting distinct
eigenvalues $\phi_{i}\gamma_{j}\neq\phi_{i}\gamma_{j^{\prime}}$ where $j\neq
j^{\prime}$;

\item $m\left(  F,G:\psi\right)  $ = probability of getting distinct
eigenvalues $\phi_{i}\gamma_{j}\neq\phi_{i^{\prime}}\gamma_{j^{\prime}}$ where
$i\neq i^{\prime}$ and $j\neq j^{\prime}$.
\end{itemize}

All the quantum logical entropies have been defined by the general method
using the information subspaces, but in the first three cases $h\left(
F:\psi\right)  $, $h\left(  G:\psi\right)  $, and $h\left(  F,G:\psi\right)
$, the density matrix method of defining logical entropies could also be used.
Then the fundamental theorem could be applied relating the quantum logical
entropies to the zeroed entities in the density matrices indicating the
eigenstates distinguished by the measurements.

The previous set identities for disjoint unions now become subspace identities
for direct sums such as:

\begin{center}
$\left[  qudit\left(  \pi\right)  \cup qudit\left(  \sigma\right)  \right]
=\left[  qudit\left(  \pi\right)  -qudit\left(  \sigma\right)  \right]
\oplus\left[  qudit\left(  \pi\right)  \cap qudit\left(  \sigma\right)
\right]  \oplus\left[  qudit\left(  \sigma\right)  -qudit\left(  \pi\right)
\right]  $.
\end{center}

Hence the probabilities are additive on those subspaces:

\begin{center}
$h\left(  F,G:\psi\right)  =h\left(  F|G:\psi\right)  +m\left(  F,G:\psi
\right)  +h\left(  G|F:\psi\right)  $.%

\begin{center}
\includegraphics[
height=1.9744in,
width=2.2736in
]%
{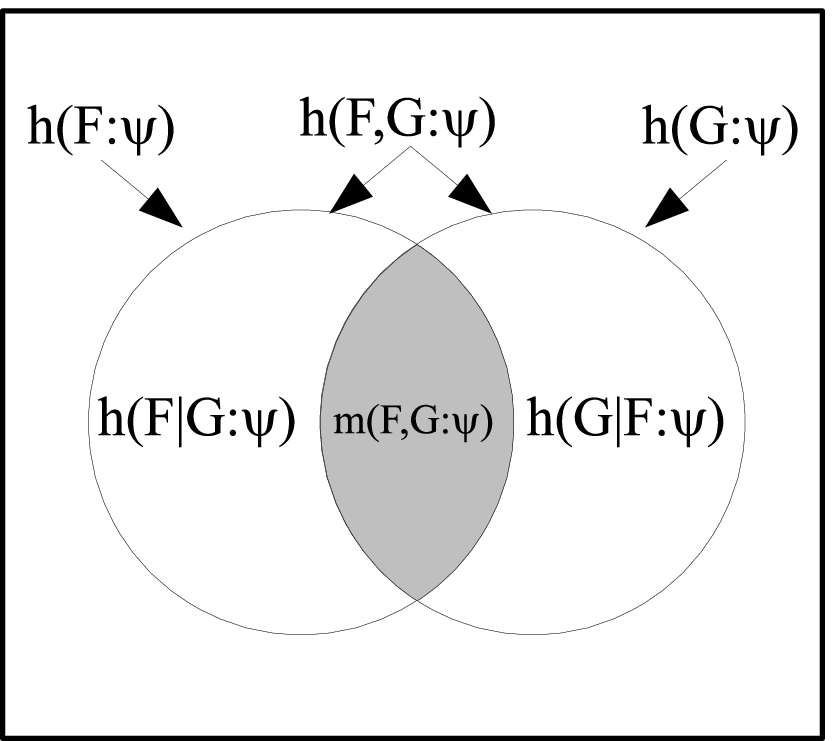}%
\end{center}

Figure 5: Venn diagram for quantum logical entropies

as values of a probability measure on $\left(  V\otimes V\right)  ^{2}$.
\end{center}

\subsection{Quantum logical entropies of density matrices in general}

The extension of the classical logical entropy $h\left(  p\right)
=1-\sum_{i=1}^{n}p_{i}^{2}$ of a probability distribution $p=\left(
p_{1},...,p_{n}\right)  $ to the quantum case is $h\left(  \rho\right)
=1-\operatorname*{tr}\left[  \rho^{2}\right]  $ where a density matrix $\rho$
replaces the probability distribution $p$ and the trace replaces the
summation. In the previous section, quantum logical entropies were defined in
terms of given observables $F,G:V\rightarrow V$ (as self-adjoint operators) as
well as a state $\psi$ and its density matrix $\rho\left(  \psi\right)  $. An
arbitrary density operator $\rho$, representing a pure or mixed state on $V$,
is also a self-adjoint operator on $V$ so quantum logical entropies can be
defined where density operators play the double role of providing the
measurement basis (as self-adjoint operators) as well as the state being measured.

Let $\rho$ and $\tau$ be two non-commuting density operators on $V$. Let
$X=\left\{  u_{i}\right\}  _{i=1,...,n}$ be an orthonormal (ON) basis of
$\rho$ eigenvectors and let $\left\{  \lambda_{i}\right\}  _{i=1,...,n}$ be
the corresponding eigenvalues which must be non-negative and sum to $1$ so
they can be interpreted as probabilities. Let $Y=\left\{  v_{j}\right\}
_{j=1,...,n}$ be an ON basis of eigenvectors for $\tau$ and let $\left\{
\mu_{j}\right\}  _{j=1,...,n}$ be the corresponding eigenvalues which are also
non-negative and sum to $1$.

Each density operator plays a double role. For instance, $\rho$ acts as the
observable to supply the measurement basis of $\left\{  u_{i}\right\}  _{i}$
and the eigenvalues $\left\{  \lambda_{i}\right\}  _{i}$ as well as being the
state to be measured supplying the probabilities $\left\{  \lambda
_{i}\right\}  _{i}$ for the measurement outcomes. Hence we could define
quantum logical entropies as in the previous section. That analysis would
analyze the distinctions between probabilities $\lambda_{i}\neq\lambda
_{i^{\prime}}$ since they are the eigenvalues too. But that analysis would not
give the quantum analogue of $h\left(  p\right)  =1-\sum_{i}p_{i}^{2}$ which
in effect uses the discrete partition on the index set $\left\{
1,...,n\right\}  $ and pays no attention to when the probabilities of
different indices are the same or different. Hence we will now develop the
analysis as in the last section but by using the discrete partition
$\mathbf{1}_{X}$ on the set of `index' states $X=\left\{  u_{i}\right\}  _{i}$
and similarly for the discrete partition $\mathbf{1}_{Y}$ on $Y=\left\{
v_{j}\right\}  _{j}$, the ON basis of eigenvectors for $\tau$.

The qudit sets of $\left(  V\otimes V\right)  \otimes\left(  V\otimes
V\right)  $ are then defined according to the identity and difference on the
index sets and independent of the eigenvalue-probabilities, e.g.,
$qudit\left(  \mathbf{1}_{X}\right)  =\left\{  \left(  u_{i}\otimes
v_{j}\right)  \otimes\left(  u_{i^{\prime}}\otimes v_{j^{\prime}}\right)
:i\neq i^{\prime}\right\}  $. Then the qudit subspaces are the subspaces of
$\left(  V\otimes V\right)  ^{2}$ generated by the qudit sets of generators:

\begin{itemize}
\item $\left[  qudit\left(  \mathbf{1}_{X}\right)  \right]  =\left[  \left(
u_{i}\otimes v_{j}\right)  \otimes\left(  u_{i^{\prime}}\otimes v_{j^{\prime}%
}\right)  :i\neq i^{\prime}\right]  $;

\item $\left[  qudit\left(  \mathbf{1}_{Y}\right)  \right]  =\left[  \left(
u_{i}\otimes v_{j}\right)  \otimes\left(  u_{i^{\prime}}\otimes v_{j^{\prime}%
}\right)  :j\neq j^{\prime}\right]  $;

\item $\left[  qudit\left(  \mathbf{1}_{X},\mathbf{1}_{Y}\right)  \right]
=\left[  qudit\left(  \mathbf{1}_{X}\right)  \cup qudit\left(  \mathbf{1}%
_{Y}\right)  \right]  =\left[  \left(  u_{i}\otimes v_{j}\right)
\otimes\left(  u_{i^{\prime}}\otimes v_{j^{\prime}}\right)  :i\neq i^{\prime
}\text{ or }j\neq j^{\prime}\right]  $;

\item $\left[  qudit\left(  \mathbf{1}_{X}|\mathbf{1}_{Y}\right)  \right]
=\left[  qudit\left(  \mathbf{1}_{X}\right)  -qudit\left(  \mathbf{1}%
_{Y}\right)  \right]  =\left[  \left(  u_{i}\otimes v_{j}\right)
\otimes\left(  u_{i^{\prime}}\otimes v_{j^{\prime}}\right)  :i\neq i^{\prime
}\text{ and }j=j^{\prime}\right]  $;

\item $\left[  qudit\left(  \mathbf{1}_{Y}|\mathbf{1}_{X}\right)  \right]
=\left[  qudit\left(  \mathbf{1}_{Y}\right)  -qudit\left(  \mathbf{1}%
_{X}\right)  \right]  =\left[  \left(  u_{i}\otimes v_{j}\right)
\otimes\left(  u_{i^{\prime}}\otimes v_{j^{\prime}}\right)  :i=i^{\prime
}\text{ and }j\neq j^{\prime}\right]  $; and

\item $\left[  qudit\left(  \mathbf{1}_{Y}\&\mathbf{1}_{X}\right)  \right]
=\left[  qudit\left(  \mathbf{1}_{Y}\right)  \cap qudit\left(  \mathbf{1}%
_{X}\right)  \right]  =\left[  \left(  u_{i}\otimes v_{j}\right)
\otimes\left(  u_{i^{\prime}}\otimes v_{j^{\prime}}\right)  :i\neq i^{\prime
}\text{ and }j\neq j^{\prime}\right]  $.
\end{itemize}

\noindent Then as qudit sets: $qudit\left(  \mathbf{1}_{X},\mathbf{1}%
_{Y}\right)  =qudit\left(  \mathbf{1}_{X}|\mathbf{1}_{Y}\right)  \uplus
qudit\left(  \mathbf{1}_{Y}|\mathbf{1}_{X}\right)  \uplus qudit\left(
\mathbf{1}_{Y}\&\mathbf{1}_{X}\right)  $, and the corresponding qudit
subspaces stand in the same relation where the disjoint union is replaced by
the disjoint sum.

The density operator $\rho$ is represented by the diagonal density matrix
$\rho_{X}$ in its own ON basis $X$ with $\left(  \rho_{X}\right)
_{ii}=\lambda_{i}$ and similarly for the diagonal density matrix $\tau_{Y}$
with $\left(  \tau_{Y}\right)  _{jj}=\mu_{j}$. The density operators
$\rho,\tau$ on $V$ define a density operator $\rho\otimes\tau$ on $V\otimes V$
with the ON basis of eigenvectors $\left\{  u_{i}\otimes v_{j}\right\}
_{i,j}$ and the eigenvalue-probabilities of $\left\{  \lambda_{i}\mu
_{j}\right\}  _{i,j}$. The operator $\rho\otimes\tau$ is represented in its ON
basis by the diagonal density matrix $\rho_{X}\otimes\tau_{Y}$ with diagonal
entries $\lambda_{i}\mu_{j}$ where $1=\left(  \lambda_{1}+...+\lambda
_{n}\right)  \left(  \mu_{1}+...+\mu_{n}\right)  =\sum_{i,j=1}^{n}\lambda
_{i}\mu_{j}$. The probability measure $p\left(  u_{i}\otimes v_{j}\right)
=\lambda_{i}\mu_{j}$ on $V\otimes V$ defines the product measure $p\times p$
on $\left(  V\otimes V\right)  ^{2}$ where it can be applied to the qudit
subspaces to define the quantum logical entropies as usual.

In the first instance, we have:

\begin{center}
$h\left(  \mathbf{1}_{X}:\rho\otimes\tau\right)  =p\times p\left(  \left[
qudit\left(  \mathbf{1}_{X}\right)  \right]  \right)  =\sum\left\{
\lambda_{i}\mu_{j}\lambda_{i^{\prime}}\mu_{j^{\prime}}:i\neq i^{\prime
}\right\}  $

$=\sum_{i\neq i^{\prime}}\lambda_{i}\lambda_{i^{\prime}}\sum_{j,j^{\prime}}%
\mu_{j}\mu_{j^{\prime}}=\sum_{i\neq i^{\prime}}\lambda_{i}\lambda_{i^{\prime}%
}=1-\sum_{i}\lambda_{i}^{2}=1-\operatorname*{tr}\left[  \rho^{2}\right]
=h\left(  \rho\right)  $
\end{center}

\noindent and similarly $h\left(  \mathbf{1}_{Y}:\rho\otimes\tau\right)
=h\left(  \tau\right)  $. Since all the data is supplied by the two density
operators, we can use simplified notation to define the corresponding joint,
conditional, and mutual entropies:

\begin{itemize}
\item $h\left(  \rho,\tau\right)  =h\left(  \mathbf{1}_{X},\mathbf{1}_{Y}%
:\rho\otimes\tau\right)  =p\times p\left(  \left[  qudit\left(  \mathbf{1}%
_{X}\right)  \cup qudit\left(  \mathbf{1}_{Y}\right)  \right]  \right)  $;

\item $h\left(  \rho|\tau\right)  =h\left(  \mathbf{1}_{X}|\mathbf{1}_{Y}%
:\rho\otimes\tau\right)  =p\times p\left(  \left[  qudit\left(  \mathbf{1}%
_{X}\right)  -qudit\left(  \mathbf{1}_{Y}\right)  \right]  \right)  $;

\item $h\left(  \tau|\rho\right)  =h\left(  \mathbf{1}_{Y}|\mathbf{1}_{X}%
:\rho\otimes\tau\right)  =p\times p\left(  \left[  qudit\left(  \mathbf{1}%
_{Y}\right)  -qudit\left(  \mathbf{1}_{X}\right)  \right]  \right)  $; and

\item $m\left(  \rho,\tau\right)  =h\left(  \mathbf{1}_{Y}\&\mathbf{1}%
_{X}:\rho\otimes\tau\right)  =p\times p\left(  \left[  qudit\left(
\mathbf{1}_{Y}\right)  \cap qudit\left(  \mathbf{1}_{X}\right)  \right]
\right)  $.
\end{itemize}

\noindent Then the usual Venn diagram relationships hold for the probability
measure $p\times p$ on $\left(  V\otimes V\right)  ^{2}$, e.g.,

\begin{center}
$h\left(  \rho,\tau\right)  =h\left(  \rho|\tau\right)  +h\left(  \tau
|\rho\right)  +m\left(  \rho,\tau\right)  $,
\end{center}

\noindent and probability interpretations are readily available. There are two
probability distributions $\lambda=\left\{  \lambda_{i}\right\}  _{i}$ and
$\mu=\left\{  \mu_{j}\right\}  _{j}$ on the sample space $\left\{
1,...,n\right\}  $. Two pairs $\left(  i,j\right)  $ and $\left(  i^{\prime
},j^{\prime}\right)  $ are drawn with replacement, the first entry in each
pair is drawn according to $\lambda$ and the second according to $\mu$. Then
$h\left(  \rho,\tau\right)  $ is the probability that $i\neq i^{\prime}$ or
$j\neq j^{\prime}$ (or both); $h\left(  \rho|\tau\right)  $ is the probability
that $i\neq i^{\prime}$ and $j=j^{\prime}$; and so forth. Note that this
interpretation assumes no special significance to a $\lambda_{i}$ and $\mu
_{i}$ having the same index since we are drawing a pair of pairs.

In the classical case of two probability distributions $p=\left(
p_{1},...,p_{n}\right)  $ and $q=\left(  q_{1},...,q_{n}\right)  $ on the same
index set, the \textit{logical cross-entropy} is defined as: $h\left(
p||q\right)  =1-\sum_{i}p_{i}q_{i}$, and interpreted as the probability of
getting different indices in drawing a single pair, one from $p$ and the other
from $q$. However, this cross-entropy assumes some special significance to
$p_{i}$ and $q_{i}$ having the same index. But in our current quantum setting,
there is no correlation between the two sets of `index' states $\left\{
u_{i}\right\}  _{i=1,...,n}$ and $\left\{  v_{j}\right\}  _{j=1,...,n}$. But
when the two density operators commute, $\tau\rho=\rho\tau$, then we can take
$\left\{  u_{i}\right\}  _{i=1,...,n}$ as an ON basis of simultaneous
eigenvectors for the two operators with respective eigenvalues $\lambda_{i}$
and $\mu_{i}$ for $u_{i}$ with $i=1,...,n$. In that special case, we can
meaningfully define the \textit{quantum logical cross-entropy} as $h\left(
\rho||\tau\right)  =1-\sum_{i=1}^{n}\lambda_{i}\mu_{i}$, but the general case
awaits further analysis below.

\section{The logical Hamming distance between two partitions}

The development of logical quantum information theory in terms of some given
commuting or non-commuting observables gives an analysis of the
distinguishability of quantum states using those observables. Without any
given observables, there is still a natural logical analysis of the distance
between quantum states that generalizes the `classical' logical distance
between partitions on a set. In the classical case, we have the logical
entropy $h\left(  \pi\right)  $ of a partition where the partition plays the
role of the direct-sum decomposition of eigenspaces of an observable in the
quantum case. But we also have just the logical entropy $h\left(  p\right)  $
of a probability distribution $p=\left(  p_{1},...,p_{n}\right)  $ and the
related compound notions of logical entropy given another probability
distribution $q=\left(  q_{1},...,q_{n}\right)  $ indexed by the same set.
This section deals with the quantum version of that situation where we are
just given density matrices such as $\rho$ and $\tau$.

First we review that classical treatment to motivate the quantum version of
the logical distance between states. A binary relation $R\subseteq U\times U$
on $U=\left\{  u_{1},...,u_{n}\right\}  $ can be represented by an $n\times n$
\textit{incidence matrix} $I(R)$ where

\begin{center}
$I\left(  R\right)  _{ij}=\left\{
\begin{array}
[c]{c}%
1\text{ if }\left(  u_{i},u_{j}\right)  \in R\\
0\text{ if }\left(  u_{i},u_{j}\right)  \notin R\text{.}%
\end{array}
\right.  $
\end{center}

\noindent Taking $R$ as the equivalence relation $\operatorname*{indit}\left(
\pi\right)  $ associated with a partition $\pi=\left\{  B_{1},...,B_{I}%
\right\}  $, the\textit{ density matrix }$\rho\left(  \pi\right)  $\textit{ of
the partition }$\pi$ (with equiprobable points) is just the incidence matrix
$I\left(  \operatorname*{indit}\left(  \pi\right)  \right)  $ rescaled to be
of trace $1$ (i.e., sum of diagonal entries is $1$):

\begin{center}
$\rho\left(  \pi\right)  =\frac{1}{\left\vert U\right\vert }I\left(
\operatorname*{indit}\left(  \pi\right)  \right)  $.
\end{center}

From coding theory \cite[p. 66]{mceliece:info}, we have the notion of the
\textit{Hamming distance between two }$0,1$\textit{ vectors or matrices }(of
the same dimensions) which is the number of places where they differ. The
powerset $\wp\left(  U\times U\right)  $ can be viewed as a vector space over
$%
\mathbb{Z}
_{2}$ where the sum of two binary relations $R,R^{\prime}\subseteq U\times U$,
symbolized $R\Delta R^{\prime}=\left(  R-R^{\prime}\right)  \cup\left(
R^{\prime}-R\right)  =R\cup R^{\prime}-R\cap R^{\prime}$, is the set of
elements (i.e., ordered pairs $\left(  u_{i},u_{j}\right)  \in U\times U$)
that are in one set or the other but not both. Thus the Hamming distance
$D_{H}\left(  I\left(  R\right)  ,I\left(  R^{\prime}\right)  \right)  $
between the incidence matrices of two binary relations is just the cardinality
of their symmetric difference: $D_{H}\left(  I\left(  R\right)  ,I\left(
R^{\prime}\right)  \right)  =\left\vert R\Delta R^{\prime}\right\vert $.
Moreover, the size of the symmetric difference does not change if the binary
relations are replaced by their complements: $\left\vert R\Delta R^{\prime
}\right\vert =\left\vert \left(  U^{2}-R\right)  \Delta\left(  U^{2}%
-R^{\prime}\right)  \right\vert $.

Hence given two partitions $\pi=\left\{  B_{1},...,B_{I}\right\}  $ and
$\sigma=\left\{  C_{1},...,C_{J}\right\}  $ on $U$, the unnormalized Hamming
distance between the two partitions is naturally defined as:\footnote{This is
investigated in Rossi \cite{rossi:part-dist}.}

\begin{center}
$D\left(  \pi,\sigma\right)  =D_{H}\left(  I\left(  \operatorname*{indit}%
\left(  \pi\right)  \right)  ,I\left(  \operatorname*{indit}\left(
\sigma\right)  \right)  \right)  =\left\vert \operatorname*{indit}\left(
\pi\right)  \Delta\operatorname*{indit}\left(  \sigma\right)  \right\vert
=\left\vert \operatorname*{dit}\left(  \pi\right)  \Delta\operatorname*{dit}%
\left(  \sigma\right)  \right\vert $,
\end{center}

\noindent and the \textit{Hamming distance between }$\pi$ and $\sigma$ is
defined as the normalized $D\left(  \pi,\sigma\right)  $:

\begin{center}
$d\left(  \pi,\sigma\right)  =\frac{D\left(  \pi,\sigma\right)  }{\left\vert
U\times U\right\vert }=\frac{\left\vert \operatorname*{dit}\left(  \pi\right)
\Delta\operatorname*{dit}\left(  \sigma\right)  \right\vert }{\left\vert
U\times U\right\vert }=\frac{\left\vert \operatorname*{dit}\left(  \pi\right)
-\operatorname*{dit}\left(  \sigma\right)  \right\vert }{\left\vert U\times
U\right\vert }+\frac{\left\vert \operatorname*{dit}\left(  \sigma\right)
-\operatorname*{dit}\left(  \pi\right)  \right\vert }{\left\vert U\times
U\right\vert }=h\left(  \pi|\sigma\right)  +h\left(  \sigma|\pi\right)  $.
\end{center}

\noindent This motivates the general case of point probabilities $p=\left(
p_{1},...,p_{n}\right)  $ where we define the \textit{Hamming distance}
between the two partitions as the sum of the two logical conditional entropies:

\begin{center}
$d\left(  \pi,\sigma\right)  =h\left(  \pi|\sigma\right)  +h\left(  \sigma
|\pi\right)  =2h(\pi\vee\sigma)-h\left(  \pi\right)  -h\left(  \sigma\right)
$ .
\end{center}

To motivate the bridge to the quantum version of the Hamming distance, we need
to calculate it using the density matrices $\rho\left(  \pi\right)  $ and
$\rho\left(  \sigma\right)  $ of the two partitions. To compute the trace
$\operatorname*{tr}\left[  \rho\left(  \pi\right)  \rho\left(  \sigma\right)
\right]  $, we compute the diagonal elements in the product $\rho\left(
\pi\right)  \rho\left(  \sigma\right)  $ and add them up: $\left[  \rho\left(
\pi\right)  \rho\left(  \sigma\right)  \right]  _{kk}=\sum_{l}\rho\left(
\pi\right)  _{kl}\rho\left(  \sigma\right)  _{lk}=\sum_{l}\sqrt{p_{k}p_{l}%
}\sqrt{p_{l}p_{k}}$where the only nonzero terms are where $u_{k},u_{l}\in
B\cap C$ for some $B\in\pi$ and $C\in\sigma$. Thus if $u_{k}\in B\cap C$, then
$\left[  \rho\left(  \pi\right)  \rho\left(  \sigma\right)  \right]
_{kk}=\sum_{u_{l}\in B\cap C}p_{k}p_{l}$. So the diagonal element for $u_{k}$
is the sum of the $p_{k}p_{l}$ for $u_{l}$ in the same intersection $B\cap C$
as $u_{k}$ so it is $p_{k}\Pr\left(  B\cap C\right)  $. Then when we sum over
the diagonal elements, then for all the $u_{k}\in B\cap C$ for any given
$B,C$, we just sum $\sum_{u_{k}\in B\cap C}p_{k}\Pr\left(  B\cap C\right)
=\Pr\left(  B\cap C\right)  ^{2}$ so that $\operatorname*{tr}\left[
\rho\left(  \pi\right)  \rho\left(  \sigma\right)  \right]  =\sum_{B\in
\pi,C\in\sigma}\Pr\left(  B\cap C\right)  ^{2}=1-h\left(  \pi\vee
\sigma\right)  $.

Hence if we define the \textit{logical cross-entropy} of $\pi$ and $\sigma$ as:

\begin{center}
$h(\pi||\sigma)=1-\operatorname*{tr}\left[  \rho\left(  \pi\right)
\rho\left(  \sigma\right)  \right]  $,
\end{center}

\noindent then for partitions on $U$ with the point probabilities $p=\left(
p_{1},...,p_{n}\right)  $, the logical cross-entropy $h\left(  \pi
||\sigma\right)  $ of two partitions is the same as the logical joint entropy
which is also the logical entropy of the join:

\begin{center}
$h\left(  \pi||\sigma\right)  =h\left(  \pi,\sigma\right)  =h\left(  \pi
\vee\sigma\right)  $.
\end{center}

\noindent Thus we can also express the logical Hamming distance between two
partitions entirely in terms of density matrices:

\begin{center}
$d\left(  \pi,\sigma\right)  =2h\left(  \pi||\sigma\right)  -h\left(
\pi\right)  -h\left(  \sigma\right)  =\operatorname*{tr}\left[  \rho\left(
\pi\right)  ^{2}\right]  +\operatorname*{tr}\left[  \rho\left(  \sigma\right)
^{2}\right]  -2\operatorname*{tr}\left[  \rho\left(  \pi\right)  \rho\left(
\sigma\right)  \right]  $.
\end{center}

\section{The quantum logical Hamming distance}

The quantum logical entropy $h\left(  \rho\right)  =1-\operatorname*{tr}%
\left[  \rho^{2}\right]  $ of a density matrix $\rho$ generalizes the
classical $h\left(  p\right)  =1-\sum_{i}p_{i}^{2}$ for a probability
distribution $p=\left(  p_{1},\ldots,p_{n}\right)  $. As a self-adjoint
operator, a density matrix has a spectral decomposition $\rho=\sum_{i=1}%
^{n}\lambda_{i}\left\vert u_{i}\right\rangle \left\langle u_{i}\right\vert $
where $\left\{  \left\vert u_{i}\right\rangle \right\}  _{i=1,...,n}$ is an
orthonormal basis for $V$ and where all the eigenvalues $\lambda_{i}$ are
real, non-negative, and $\sum_{i=1}^{n}\lambda_{i}=1$. Then $h\left(
\rho\right)  =1-\sum_{i}\lambda_{i}^{2}$ so $h\left(  \rho\right)  $ can be
interpreted as the probability of getting distinct indices $i\neq i^{\prime}$
in two independent measurements of the state $\rho$ with $\left\{  \left\vert
u_{i}\right\rangle \right\}  $ as the measurement basis. Classically, it is
the two-draw probability of getting distinct indices in two independent
samples of the probability distribution $\lambda=\left(  \lambda_{1}%
,\ldots,\lambda_{n}\right)  $, just as $h\left(  p\right)  $ is the
probability of getting distinct indices in two independent draws on $p$. For a
pure state $\rho$, there is one $\lambda_{i}=1$ with the others zero, and
$h\left(  \rho\right)  =0$ is the probability of getting distinct indices in
two independent draws on $\lambda=\left(  0,\ldots,0,1,0,\ldots,0\right)  $.

In the classical case of the logical entropies, we worked with the ditsets or
sets of distinctions of partitions. But everything could also be expressed in
terms of the complementary sets of indits or indistinctions of partitions
(ordered pairs of elements in the same block of the partition) since:
$\operatorname*{dit}\left(  \pi\right)  \uplus\operatorname*{indit}\left(
\pi\right)  =U\times U$. When we switch to the density matrix treatment of
`classical' partitions, then the focus shifts to the indistinctions. For a
partition $\pi=\left\{  B_{1},\ldots,B_{I}\right\}  $, the logical entropy is
the sum of the distinction probabilities: $h\left(  \pi\right)  =\sum_{\left(
u_{k},u_{l}\right)  \in\operatorname*{dit}\left(  \pi\right)  }p_{k}p_{l}$
which in terms of indistinctions is:

\begin{center}
$h\left(  \pi\right)  =1-\sum_{\left(  u_{k},u_{l}\right)  \in
\operatorname*{indit}\left(  \pi\right)  }p_{k}p_{l}=1-\sum_{i=1}^{I}%
\Pr\left(  B_{i}\right)  ^{2}$.
\end{center}

\noindent When expressed in the density matrix formulation, then
$\operatorname*{tr}\left[  \rho\left(  \pi\right)  ^{2}\right]  $ is the sum
of the indistinction probabilities:

\begin{center}
$\operatorname*{tr}\left[  \rho\left(  \pi\right)  ^{2}\right]  =\sum_{\left(
u_{k},u_{l}\right)  \in\operatorname*{indit}\left(  \pi\right)  }p_{k}%
p_{l}=\sum_{i=1}^{I}\Pr\left(  B_{i}\right)  ^{2}$.
\end{center}

\noindent The nonzero entries in $\rho\left(  \pi\right)  $ have the form
$\sqrt{p_{k}p_{l}}$ for $\left(  u_{k},u_{l}\right)  \in\operatorname*{indit}%
\left(  \pi\right)  $; their squares are the indistinction probabilities. That
provides the interpretive bridge to the quantum case.

The quantum analogue of an indistinction probability is the absolute square
$\left\vert \rho_{kl}\right\vert ^{2}$ of a nonzero entry $\rho_{kl}$ in a
density matrix $\rho$ and $\operatorname*{tr}\left[  \rho^{2}\right]
=\sum_{k,l}\left\vert \rho_{kl}\right\vert ^{2}$ is the sum of those
`indistinction' probabilities. The nonzero entries in the density matrix
$\rho$ might be called \textquotedblleft\textit{coherences}\textquotedblright%
\ so that $\rho_{kl}$ may be interpreted as the amplitudes for the states
$u_{k}$ and $u_{l}$ to cohere together in the state $\rho$ so
$\operatorname*{tr}\left[  \rho^{2}\right]  $ is the sum of the
\textit{coherence probabilities}--just as $\operatorname*{tr}\left[
\rho\left(  \pi\right)  ^{2}\right]  =\sum_{\left(  u_{k},u_{l}\right)
\in\operatorname*{indit}\left(  \pi\right)  }p_{k}p_{l}$ is the sum of the
indistinction probabilities. The quantum logical entropy $h\left(
\rho\right)  =1-\operatorname*{tr}\left[  \rho^{2}\right]  $ may then be
interpreted as the sum of the \textit{decoherence probabilities}--just as
$h\left(  \rho\left(  \pi\right)  \right)  =h\left(  \pi\right)
=1-\sum_{\left(  u_{k},u_{l}\right)  \in\operatorname*{indit}\left(
\pi\right)  }p_{k}p_{l}$ is the sum of the distinction probabilities.

The general quantum version of the joint entropy $h\left(  \pi,\sigma\right)
=h\left(  \pi\vee\sigma\right)  =h\left(  \pi||\sigma\right)  $ is the:

\begin{center}
$h\left(  \rho||\tau\right)  =1-\operatorname*{tr}\left[  \tau\rho\right]  $

\textit{quantum logical cross-entropy}.
\end{center}

To work out its interpretation, we again take ON eigenvector bases $\left\{
\left\vert u_{i}\right\rangle \right\}  _{i=1}^{n}$ for $\rho$ and $\left\{
\left\vert v_{j}\right\rangle \right\}  _{j=1}^{n}$ for $\tau$ with
$\lambda_{i}$ and $\mu_{j}$ as the respective eigenvalues, and compute the
operation of $\tau\rho:V\rightarrow V$. Now $\left\vert u_{i}\right\rangle
=\sum_{j}\left\langle v_{j}|u_{i}\right\rangle \left\vert v_{j}\right\rangle $
so $\rho\left\vert u_{i}\right\rangle =\lambda_{i}\left\vert u_{i}%
\right\rangle =\sum_{j}\lambda_{i}\left\langle v_{j}|u_{i}\right\rangle
\left\vert v_{j}\right\rangle $ and then $\tau\rho\left\vert u_{i}%
\right\rangle =\sum_{j}\lambda_{i}\mu_{j}\left\langle v_{j}|u_{i}\right\rangle
\left\vert v_{j}\right\rangle $. Thus $\tau\rho$ in the $\left\{
u_{i}\right\}  _{i}$ basis would have the diagonal entries $\left\langle
u_{i}|\tau\rho|u_{i}\right\rangle =\sum_{j}\lambda_{i}\mu_{j}\left\langle
v_{j}|u_{i}\right\rangle \left\langle u_{i}|v_{j}\right\rangle $ so the trace is:

\begin{center}
$\operatorname*{tr}\left[  \tau\rho\right]  =\sum_{i}\left\langle u_{i}%
|\tau\rho|u_{i}\right\rangle =\sum_{i,j}\lambda_{i}\mu_{j}\left\langle
v_{j}|u_{i}\right\rangle \left\langle u_{i}|v_{j}\right\rangle
=\operatorname*{tr}\left[  \rho\tau\right]  $
\end{center}

\noindent which is symmetrical. The other information we have is the $\sum
_{i}\lambda_{i}=1=\sum_{j}\mu_{j}$ and they are non-negative. The classical
logical cross-entropy of two probability distributions is $h\left(
p||q\right)  =1-\sum_{i}p_{i}q_{i}$ where two indices $i$ and $i^{\prime}$ are
either identical or totally distinct. But in the quantum case, the `index'
states $\left\vert u_{i}\right\rangle $ and $\left\vert v_{j}\right\rangle $
have an `overlap' measured by the inner product $\left\langle u_{i}%
|v_{j}\right\rangle $. The trace $\operatorname*{tr}\left[  \rho\tau\right]  $
is real since $\left\langle v_{j}|u_{i}\right\rangle =\left\langle u_{i}%
|v_{j}\right\rangle ^{\ast}$ and $\left\langle v_{j}|u_{i}\right\rangle
\left\langle u_{i}|v_{j}\right\rangle =\left\vert \left\langle u_{i}%
|v_{j}\right\rangle \right\vert ^{2}=\left\vert \left\langle v_{j}%
|u_{i}\right\rangle \right\vert ^{2}$ is the probability of getting
$\lambda_{i}$ when measuring $v_{j}$ in the $u_{i}$ basis and the probability
of getting $\mu_{j}$ when measuring $u_{i}$ in the $v_{j}$ basis. The twofold
nature of density matrices as states and as observables then allows
$\operatorname*{tr}\left[  \rho\tau\right]  $ to be interpreted as the average
value of the observable $\rho$ when measuring the state $\tau$ or vice-versa.

We may call $\left\langle v_{j}|u_{i}\right\rangle \left\langle u_{i}%
|v_{j}\right\rangle $ the \textit{proportion }or\textit{ extent of overlap}
for those two index states. Thus $\operatorname*{tr}\left[  \rho\tau\right]  $
is the sum of \textit{all} the probability combinations $\lambda_{i}\mu_{j}$
weighted by the overlaps $\left\langle v_{j}|u_{i}\right\rangle \left\langle
u_{i}|v_{j}\right\rangle $ for the index states $\left\vert u_{i}\right\rangle
$ and $\left\vert v_{j}\right\rangle $. The quantum logical cross-entropy can
be written in a number of ways:

\begin{center}
$h\left(  \rho||\tau\right)  =1-\operatorname*{tr}\left[  \rho\tau\right]
=1-\sum_{i,j}\lambda_{i}\mu_{j}\left\langle v_{j}|u_{i}\right\rangle
\left\langle u_{i}|v_{j}\right\rangle $

$=\operatorname*{tr}\left[  \tau\left(  I-\rho\right)  \right]  =\sum
_{i,j}\left(  1-\lambda_{i}\right)  \mu_{j}\left\langle v_{j}|u_{i}%
\right\rangle \left\langle u_{i}|v_{j}\right\rangle $

$=\operatorname*{tr}\left[  \rho\left(  I-\tau\right)  \right]  =\sum
_{i,j}\lambda_{i}\left(  1-\mu_{j}\right)  \left\langle v_{j}|u_{i}%
\right\rangle \left\langle u_{i}|v_{j}\right\rangle $.
\end{center}

\noindent Classically, the `index state' $\left\{  i\right\}  $ completely
overlaps with $\left\{  j\right\}  $ when $i=j$ and has no overlap with any
other $\left\{  i^{\prime}\right\}  $ from the indices $\left\{
1,\ldots,n\right\}  $ so the `overlaps' are, as it were, $\left\langle
j|i\right\rangle \left\langle i|j\right\rangle =\delta_{ij}$, the Kronecker
delta. Hence the classical analogue formulas are:

\begin{center}
$h\left(  p||q\right)  =1-\sum_{i,j}p_{i}q_{j}\delta_{ij}=\sum_{i,j}\left(
1-p_{i}\right)  q_{j}\delta_{ij}=\sum_{i,j}p_{i}\left(  1-q_{j}\right)
\delta_{ij}$.
\end{center}

The quantum logical cross-entropy $h\left(  \rho||\tau\right)  $ can be
interpreted by considering two measurements, one of $\rho$ with the $\left\{
\left\vert u_{i}\right\rangle \right\}  _{i}$ measurement basis and the other
of $\tau$ with the $\left\{  \left\vert v_{j}\right\rangle \right\}  _{j}$
measurement basis. If the outcome of the $\rho$ measurement was $u_{i}$ with
probability $\lambda_{i}$, then the outcome of the $\tau$ measurement is
different than $v_{j}$ with probability $1-\mu_{j}$ but that distinction
probability $\lambda_{i}\left(  1-\mu_{j}\right)  $ is only relevant to the
extent that $u_{i}$ and $v_{j}$ are the `same state' or overlap, and that
extent is $\left\langle v_{j}|u_{i}\right\rangle \left\langle u_{i}%
|v_{j}\right\rangle $. Hence the quantum logical cross-entropy is the sum of
those two-measurement distinction probabilities weighted by the extent that
the states overlap.\footnote{The interpretation of $h\left(  \rho\right)  $
and $h\left(  \tau||\rho\right)  $, as well as the later development of the
quantum logical conditional entropy $h\left(  \rho|\tau\right)  $ and the
quantum Hamming distance $d\left(  \rho,\tau\right)  $, are all based on using
the eigenvectors and eigenvalues of density matrices--which Michael Nielsen
and Issac Chuang seem to prematurely dismiss as having little or no
\textquotedblleft special significance.\textquotedblright\ \cite[p.
103]{nielsen-chuang:bible}}

When the two density matrices commute, $\rho\tau=\tau\rho$, then (as noted
above) we have the essentially classical situation of one set of index states
$\left\{  \left\vert u_{i}\right\rangle \right\}  _{i}$ which is an
orthonormal basis set of simultaneous eigenvectors for both $\rho$ and $\tau$
with the respective eigenvalues $\left\{  \lambda_{i}\right\}  _{i}$ and
$\left\{  \mu_{j}\right\}  _{j}$. Then $\left\langle u_{j}|u_{i}\right\rangle
\left\langle u_{i}|u_{j}\right\rangle =\delta_{ij}$ so $h\left(  \rho
||\tau\right)  =\sum_{i,j}\lambda_{i}\left(  1-\mu_{j}\right)  \delta_{ij}$ is
the probability of getting two distinct index states $u_{i}$ and $u_{j}$ for
$i\neq j$ in two independent measurements, one of $\rho$ and one of $\tau$ in
the same measurement basis of $\left\{  \left\vert u_{i}\right\rangle
\right\}  _{i}$. This interpretation includes the special case when $\tau
=\rho$ and $h\left(  \rho||\rho\right)  =h\left(  \rho\right)  $.

We saw that classically, the logical Hamming distance between two partitions
could be defined as:

\begin{center}
$d\left(  \pi,\sigma\right)  =2h\left(  \pi||\sigma\right)  -h\left(
\pi\right)  -h\left(  \sigma\right)  =\operatorname*{tr}\left[  \rho\left(
\pi\right)  ^{2}\right]  +\operatorname*{tr}\left[  \rho\left(  \sigma\right)
^{2}\right]  -2\operatorname*{tr}\left[  \rho\left(  \pi\right)  \rho\left(
\sigma\right)  \right]  $
\end{center}

\noindent so this motivates the quantum definition:\footnote{Nielsen and
Chuang suggest the idea of a Hamming distance between quantum states--only to
then dismiss it. \textquotedblleft Unfortunately, the Hamming distance between
two objects is simply a matter of labeling, and \textit{a priori} there aren't
any labels in the Hilbert space arena of quantum mechanics!\textquotedblright%
\ \cite[p. 399]{nielsen-chuang:bible} They are right that there is no
correlation, say, between the vectors in the two ON bases $\left\{
u_{i}\right\}  _{i}$ and $\left\{  v_{j}\right\}  _{j}$ for $V$, but the
cross-entropy $h\left(  \rho||\tau\right)  $ uses all possible combinations in
the terms $\lambda_{i}\left(  1-\mu_{j}\right)  \left\langle v_{j}%
|u_{i}\right\rangle \left\langle u_{i}|v_{j}\right\rangle $ and thus the
definition of the Hamming distance developed here does not use any arbitrary
labeling or correlations.}

\begin{center}
$d\left(  \rho,\tau\right)  =2h\left(  \rho||\tau\right)  -h\left(
\rho\right)  -h\left(  \tau\right)  =\operatorname*{tr}\left[  \rho
^{2}\right]  +\operatorname*{tr}\left[  \tau^{2}\right]  -2\operatorname*{tr}%
\left[  \rho\tau\right]  $

\textit{quantum logical Hamming distance between two quantum states}.
\end{center}

There is another distance measure between quantum states, namely the
Hilbert-Schmidt norm, that has been recently investigated in
\cite{tamir-cohen:hilbert-schmidt} (with an added $\frac{1}{2}$
factor):\footnote{It is the square of the Euclidean distance between the
quantum states and, ignoring the $\frac{1}{2}$ factor, it is the square of the
\textquotedblleft trace distance\textquotedblright\ \cite[Chapter
9]{nielsen-chuang:bible} between the states.}

\begin{center}
$\operatorname*{tr}\left[  \left(  \rho-\tau\right)  ^{2}\right]  $

Hilbert-Schmidt distance
\end{center}

\noindent where we write $A^{2}$ for $A^{\dagger}A$. Then the naturalness of
this distance is enhanced by the fact that it is the same as the quantum
Hamming distance:

\begin{center}
$\operatorname*{tr}\left[  \left(  \rho-\tau\right)  ^{2}\right]
=\operatorname*{tr}\left[  \rho^{2}\right]  +\operatorname*{tr}\left[
\tau^{2}\right]  -2\operatorname*{tr}\left[  \rho\tau\right]  =2h\left(
\rho||\tau\right)  -h\left(  \rho\right)  -h\left(  \tau\right)  =d\left(
\rho,\tau\right)  $

Hilbert-Schmidt distance = quantum logical Hamming distance between quantum states.
\end{center}

\noindent Hence the \textit{information inequality} holds trivially for the
quantum logical Hamming distance:

\begin{center}
$d\left(  \rho,\tau\right)  \geq0$ with equality iff $\rho=\tau$.
\end{center}

\section{Concluding Remarks}

Logical information theory arises as the quantitative version of the logic of
partitions just as logical probability theory arises as the quantitative
version of the dual Boolean logic of subsets. Philosophically, logical
information is based on the idea of information-as-distinctions. The Shannon
definitions of entropy arise naturally out of the logical definitions by
replacing the counting of distinctions by the counting of the minimum number
of binary partitions (bits) that are required, on average, to make all the
same distinctions, i.e., to uniquely encode the distinguished elements--which
is why the Shannon theory is so well-adapted for the theory of coding and communication.

This `classical' logical information theory may be developed with the data of
two partitions on a set with point probabilities. Section 7 gives the
generalization to the quantum case where the partitions are provided by two
commuting observables, the point set is an ON basis of simultaneous
eigenvectors, and the point probabilities are provided by the state to be
measured. In Section 8, the fundamental theorem for quantum logical entropy
and measurement established a direct quantitative connection between the
increase in quantum logical entropy due to a projective measurement and the
eigenstates (cohered together in the pure superposition state being measured)
that are distinguished by the measurement (decohered in the post-measurement
mixed state). This theorem establishes quantum logical entropy as a natural
notion for a quantum information theory focusing on distinguishing states.

The classical theory might also start with partitions on two different sets
and a probability distribution on the product of the sets. Section 9 gives the
quantum generalization of that case with the two sets being two ON bases for
two non-commuting observables, and the probabilities are provided by a state
to be measured. The classical theory may also be developed just using two
probability distributions indexed by the same set, and this is generalized to
the quantum case where we are just given two density matrices representing two
states in a Hilbert space. Sections 10 and 11 carry over the Hamming distance
measure from the classical to the quantum case where it is equal to the
Hilbert-Schmidt distance measure (square of the trace distance).

The overall argument is that quantum logical entropy is the simple and natural
notion of information-as-distinctions for quantum information theory focusing
on the distinguishing of quantum states.


\begin{thebibliography}{99}                                                                                               %


\bibitem {abramson:it}Abramson, Norman 1963. \textit{Information Theory and
Coding}. New York: McGraw-Hill.

\bibitem {auletta:qm}Auletta, Gennaro, Mauro Fortunato, and Giorgio Parisi.
2009. \textit{Quantum Mechanics}. Cambridge UK: Cambridge University Press.

\bibitem {bennett:qinfo}Bennett, Charles H. 2003. Quantum Information: Qubits
and Quantum Error Correction. \textit{International Journal of Theoretical
Physics} 42 (2 February): 153--76.

\bibitem {boole:lot}Boole, George 1854. \textit{An Investigation of the Laws
of Thought on which are founded the Mathematical Theories of Logic and
Probabilities}. Cambridge: Macmillan and Co.

\bibitem {buscemi:lin-entropy}Buscemi, Fabrizio, Paolo Bordone, and Andrea
Bertoni. 2007. Linear Entropy as an Entanglement Measure in Two-Fermion
Systems. \textit{ArXiv.org}. March 2. http://arxiv.org/abs/quant-ph/0611223v2.

\bibitem {camp:meas}Campbell, L. Lorne 1965. Entropy as a
Measure.\textit{\ IEEE Trans. on Information Theory}. IT-11 (January): 112-114.

\bibitem {cohen-t:QM1}Cohen-Tannoudji, Claude, Bernard Diu and Franck
Lalo\"{e} 2005. \textit{Quantum Mechanics Vol. 1}. New York: John Wiley \& Sons.

\bibitem {ell:countingdits}Ellerman, David. 2009. Counting Distinctions: On
the Conceptual Foundations of Shannon's Information Theory. \textit{Synthese}
168 (1 May): 119--49.

\bibitem {ell:partitions}Ellerman, David 2010. The Logic of Partitions:
Introduction to the Dual of the Logic of Subsets. \textit{Review of Symbolic
Logic}. 3 (2 June): 287-350.

\bibitem {ell:intropartlogic}Ellerman, David 2014. An Introduction of
Partition Logic. \textit{Logic Journal of the IGPL.} 22, no. 1: 94--125.

\bibitem {ell:lit}Ellerman, David. 2017. Logical Information Theory: New
Foundations for Information Theory. \textit{Logic Journal of the IGPL} 25 (5
Oct.): 806--35.

\bibitem {ell:qpl}Ellerman, David. 2018. The Quantum Logic of Direct-Sum
Decompositions: The Dual to the Quantum Logic of Subspaces. \textit{Logic
Journal of the IGPL} 26 (1 January):1--13.

\bibitem {fano:density}Fano, U. 1957. Description of States in Quantum
Mechanics by Density Matrix and Operator Techniques. \textit{Reviews of Modern
Physics} 29 (1): 74--93.

\bibitem {hart:ti}Hartley, Ralph V. L. 1928. Transmission of information.
\textit{Bell System Technical Journal}. 7 (3, July): 535-63.

\bibitem {havrda:alpha}Havrda, Jan, and Frantisek Charvat. 1967.
Quantification Methods of Classification Processes: Concept of Structural
$\alpha$-Entropy. \textit{Kybernetika} (Prague) 3: 30--35.

\bibitem {jaeger:qinfo}Jaeger, Gregg. 2007. \textit{Quantum Information: An
Overview}. New York: Springer Science+Business Media.

\bibitem {kolmogor:combfound}Kolmogorov, Andrei N. 1983. Combinatorial
Foundations of Information Theory and the Calculus of Probabilities.
\textit{Russian Math. Surveys} 38 (4): 29--40.

\bibitem {kung-rota-yan:comb}Kung, Joseph P. S., Gian-Carlo Rota, and
Catherine H. Yan. 2009. \textit{Combinatorics: The Rota Way}. New York:
Cambridge University Press.

\bibitem {law:sfm}Lawvere, F. William and Robert Rosebrugh 2003. \textit{Sets
for Mathematics}. Cambridge UK: Cambridge University Press.

\bibitem {mceliece:info}McEliece, R. J. 1977. \textit{The Theory of
Information and Coding: A Mathematical Framework for Communication
(Encyclopedia of Mathematics and Its Applications, Vol. 3)}. Reading MA: Addison-Wesley.

\bibitem {nielsen-chuang:bible}Nielsen, M., and I. Chuang. 2000.
\textit{Quantum Computation and Quantum Information}. Cambridge UK: Cambridge
University Press.

\bibitem {rao:div}Rao, C. R. 1982. Diversity and Dissimilarity Coefficients: A
Unified Approach. \textit{Theoretical Population Biology}. 21: 24-43.

\bibitem {rossi:part-dist}Rossi, Giovanni. 2011. Partition Distances.
\textit{arXiv:1106.4579v1}.

\bibitem {rota:fubini}Rota, Gian-Carlo. 2001. Twelve Problems in Probability
No One Likes to Bring up. In \textit{Algebraic Combinatorics and Computer
Science}, edited by Henry Crapo and Domenico Senato, 57--93. Milano: Springer.

\bibitem {rozeboom:partials}Rozeboom, William W. 1968. The Theory of Abstract
Partials: An Introduction.\ \textit{Psychometrika} 33 (2 June): 133--67.

\bibitem {shannon:comm}Shannon, Claude E. 1948. A Mathematical Theory of
Communication. \textit{Bell System Technical Journal}. 27: 379-423; 623-56.

\bibitem {svozil:qlogic}Svozil, Karl. 1998. \textit{Quantum Logic}. Singapore:
Springer-Verlag Singapore.

\bibitem {tamir-cohen:logicalentropy}Tamir, Boaz, and Eliahu Cohen. 2014.
Logical Entropy for Quantum States. \textit{ArXiv.org}. December. http://de.arxiv.org/abs/1412.0616v2.

\bibitem {tamir-cohen:hilbert-schmidt}Tamir, Boaz, and Eliahu Cohen. 2015. A
Holevo-Type Bound for a Hilbert Schmidt Distance Measure. \textit{Journal of
Quantum Information Science} 5: 127--33.

\bibitem {tsallis:entropy}Tsallis, Constantino 1988. Possible Generalization
for Boltzmann-Gibbs Statistics. \textit{J. Stat. Physics} 52: 479--87.

\bibitem {wright:genurnmodels}Wright, Ron. 1990. Generalized Urn Models.
\textit{Foundations of Physics} 20 (7):881--903.
\end{thebibliography}
\end{document}